\def\year{2021}\relax
\documentclass[letterpaper]{article} 
\usepackage{aaai21}  
\usepackage{times}  
\usepackage{helvet} 
\usepackage{courier}  
\usepackage[hyphens]{url}  
\usepackage{graphicx} 
\usepackage{natbib}  
\usepackage{caption} 
\makeatletter
\def\copyright@on{F}
\makeatother
\usepackage{verbatim}
\usepackage{boxedminipage}
\usepackage{adjustbox}
\usepackage{amsmath}
\usepackage{amssymb}
\usepackage{amsfonts}
\usepackage{amsthm}

\usepackage[switch]{lineno}

\usepackage{enumitem}

\newcommand{\myparagraph}[1]{\medskip \noindent\textbf{#1}}
\newcommand{\sv}[1]{}
\newcommand{\lv}[1]{#1}
\newif\iflong
\newif\ifshort

\lv{  \usepackage{enumitem}}

\iflong
\else
\shorttrue
\fi

\title{The Parameterized Complexity of Clustering Incomplete Data}
\author{%
  Eduard Eiben,\textsuperscript{\rm 1}
  Robert Ganian,\textsuperscript{\rm 2}
  Iyad Kanj,\textsuperscript{\rm 3}
  Sebastian Ordyniak,\textsuperscript{\rm 4}
  Stefan Szeider\textsuperscript{\rm 2}\\}
\affiliations {%
  \textsuperscript{\rm 1}Department of Computer Science, Royal
  Holloway,   University of London, Egham, UK\\
 \textsuperscript{\rm 2}Algorithms and Complexity Group, TU Wien, Vienna, Austria\\
 \textsuperscript{\rm 3}School of Computing, DePaul University, Chicago, USA\\
\textsuperscript{\rm 4}University of Leeds, School of Computing,
Leeds, UK\\
\{eduard.eiben,rganian\}@gmail.com, ikanj@cdm.depaul.edu, sordyniak@gmail.com, stefan@szeider.net
}

\newcommand{\DIAMq}{\DIAM}
\newcommand{\MDIAMq}{\DIAM_{\max}}

\usepackage{floatrow}

\usepackage[ruled,linesnumbered]{algorithm2e} 

\SetAlFnt{\small}
\SetAlCapFnt{\small}
\SetAlCapNameFnt{\small}
\SetAlCapHSkip{0pt}
\IncMargin{-\parindent}

\usepackage{booktabs} 

\usepackage{boxedminipage}

\usepackage[obeyFinal,colorinlistoftodos,textsize=tiny]{todonotes}

\usepackage{cleveref}

\newtheorem{theorem}{Theorem}

\newtheorem*{theorem*}{Theorem}
\newtheorem{corollary}[theorem]{Corollary}
\crefname{corollary}{corollary}{corollaries}
\newtheorem{lemma}[theorem]{Lemma}
\newtheorem{observation}[theorem]{Observation}
\crefname{lemma}{lemma}{lemmas}

\crefname{definition}{definition}{definitions}

\newcommand{\SB}{\{\,}
\newcommand{\SM}{\;{|}\;}
\newcommand{\SE}{\,\}}

\newcommand{\PPP}{\mathcal{P}}
\newcommand{\CCC}{\mathcal{C}}

\newcommand{\III}{\mathcal{I}}
\newcommand{\FFF}{\mathcal{F}}

\newcommand{\Nat}{\mathbb{N}}

\newcommand{\cc}[1]{{\mbox{\textnormal{\textsf{#1}}}}\xspace}  

\newcommand{\NP}{\cc{NP}}

\newcommand{\FPT}{\cc{FPT}}
\newcommand{\XP}{\cc{XP}}
\newcommand{\Weft}{{\cc{W}}}
\newcommand{\W}[1]{{\Weft}{{\normalfont{[#1]}}}}
\newcommand{\paraNP}{\cc{paraNP}}

\newcommand{\hy}{\hbox{-}\nobreak\hskip0pt}

\newcommand{\nn}{\mathbb{N}}

\newcommand{\bigoh}{\mathcal{O}}

\newcommand{\probfont}[1]{\textnormal{\textsc{#1}}}

\newcommand{\INCLUS}{\probfont{In-Clustering}}
\newcommand{\ANYCLUS}{\probfont{Any-Clustering}}
\newcommand{\PAIRCLUS}{\probfont{Diam-Clustering}}
\newcommand{\CLUS}[1]{\probfont{#1-Clustering}}
\newcommand{\IADCLUS}{\probfont{In/Any/Diam-Clustering}}

\newcommand{\blank}{{\small\square}}
\newcommand{\INCLUSq}{\probfont{In-Clustering-Completion}}
\newcommand{\ANYCLUSq}{\probfont{Any-Clustering-Completion}}
\newcommand{\PAIRCLUSq}{\probfont{Diam-Clustering-Completion}}
\newcommand{\CLUSq}[1]{\probfont{#1-Clustering-Completion}}
\newcommand{\IADCLUSq}{\probfont{In/Any/Diam-Clustering-Completion}}

\newcommand{\TCLUSq}[1]{\probfont{#1-Clustering-C}}

\newcommand{\comb}{\textnormal{\texttt{cover}}}

\newcommand{\yes}{\textsc{Yes}}

\def\ie{{\em i.e.}}
\def\eg{{\em e.g.}}
\def\etal{{\em et al.}}
\newcommand{\pbDef}[3]{%
\noindent
\begin{center}
\begin{boxedminipage}{0.98 \columnwidth}
#1\\[5pt]
\begin{tabular}{l p{0.73 \columnwidth}}
Input: & #2\\
Question: & #3
\end{tabular}
\end{boxedminipage}
\end{center}
}

\newcommand{\NO}{\textsc{No}}
\newcommand{\YES}{\textsc{Yes}}
\newcommand{\HSET}{\Delta}
\newcommand{\HVEC}{\HSET^{-1}}
\newcommand{\HDIST}{\delta}
\newcommand{\DIAM}{\gamma}
\newcommand{\DCOOR}{Z}
\newcommand{\HN}[2]{N_{#2}({#1})}
\newcommand{\inst}{\mathcal{I}}
\newcommand{\compG}{G}

\newcommand{\HDISTq}{\HDIST}

\begin{document}

\maketitle

\pagestyle{plain}

\begin{abstract}
  We study fundamental clustering problems for incomplete
  data. Specifically, given a set of incomplete $d$-dimensional
  vectors (representing rows of a matrix), the goal is to complete the
  missing vector entries in a way that admits a partitioning of
  the vectors into at most $k$ clusters with radius or diameter at
  most $r$.  We give tight characterizations of the parameterized
  complexity of these problems with respect to the parameters $k$,
  $r$, and the minimum number of rows and columns needed to cover all
  the missing entries. We show that the considered problems are
  fixed-parameter tractable when parameterized by the three parameters
  combined, and that dropping any of the three parameters results in
  parameterized intractability.  A byproduct of our results is that,
  for the complete data setting, all problems under consideration are
  fixed-parameter tractable parameterized by $k+r$.
\end{abstract}

\section{Introduction}
\label{sec:intro}

We study fundamental clustering problems for incomplete data. In this
setting, we are given a set of $d$-dimensional Boolean vectors
(regarded as rows of a matrix), some of whose entries
might be missing. The objective is to complete the missing entries
in order to enable a ``clustering'' of the $d$-dimensional
vectors such that  elements in the same cluster are ``similar.''

There is a wealth of research on
data completion problems~\cite{cp10,cr09,ct10,ev13,icml,hmrw14} due to
their ubiquitous applications in recommender systems, machine
learning, sensing, computer vision, data science, and predictive
analytics, among others. In these areas, data completion problems
naturally arise after observing a sample from the set of vectors, and
attempting to recover the missing entries with the goal of
optimizing certain criteria. Some of these criteria include
minimizing the number of clusters into which the completed vectors can be partitioned, or forming a large cluster, where the definition of what constitutes a cluster varies from one application to another~\cite{balzano,nips2016-incompleteclustering,ev13,icdm}.
\sv{
Needless to say, the clustering problem itself (\ie, for complete
data) is a fundamental problem whose applications
span several areas of computing, including data mining, machine
learning, pattern recognition, and recommender systems~\cite{clusteringbook1,clusteringbook2,clusteringbook4,clusteringbook3}.
}

\lv{
Needless to say, the clustering problem itself (\ie, for complete
data) is a fundamental problem whose applications
span several areas of computing, including data mining, machine
learning, pattern recognition, and recommender systems;
there are several recent
books~\cite{clusteringbook1,clusteringbook2,clusteringbook4,clusteringbook3}
that provide an introduction to clustering and its applications. Clustering is the focus of extensive research in the Machine Learning and Neural Information Processing communities, with numerous papers studying application-focused~\cite{BetancourtZMWZS16,HuCNCG18,0001SC18} as well as purely theoretical~\cite{HarrisLSTP18,AddadKM18,Ryabko17,YunP16} aspects of clustering.
The clustering problem is formulated by representing each element in the given set as a $d$-dimensional vector
 each of whose coordinates corresponds to a feature/characteristic, and the value of the vector at that coordinate reflects the score of the element with respect to that characteristic.}

In many cases, the goal of clustering is to optimize the number of clusters and/or the degree of similarity within a cluster (\emph{intra-cluster similarity}).
To measure the intra-cluster similarity, apart from using an aggregate
measure (\eg, the variance in $k$-means clustering), two measures that
have been studied use the \emph{radius} (maximum distance to a selected ``center'' vector)
and \emph{diameter} (maximum distance between any two cluster-vectors)
of the
cluster~\cite{charikar,frieze,normclustering,lingashammingcenter,lingasapxclustering,gonzalez,GrammNiedermeierRossmanith03}. The
radius is computed either with respect to a vector in the
cluster itself 
or an arbitrary
$d$-dimensional vector~\cite{clusteringbook4}.

Regardless of which of the above measures of intra-cluster similarity is used,
the vast majority of the clustering problems that arise are \NP-hard.
Consequently, heuristics are often used to cope with the hardness of
clustering problems, trading in a suboptimal clustering for polynomial
running time. In this paper we take a different approach: we maintain the optimality of the obtained clustering by
relaxing the notion of tractability from polynomial-time to
\emph{fixed-parameter tractability} (\FPT{})~\cite{CyganFKLMPPS15,DowneyFellows13,GottlobSzeider08}, where the running time is
polynomial in the instance size but may involve a super-polynomial factor that depends only on some \emph{problem parameter}, which is assumed to be small for certain instances of interest.
 In the context of clustering,
two natural parameters that are desirable to be small
are upper bounds on the number of clusters and the
radius/diameter.
\lv{
Such clusterings are suitable for many applications, as one would like the similarity level within each cluster to be high and the number of clusters not to be very large.
}

\smallskip

\noindent \textbf{Contributions.}
Motivated by the above, we consider several fundamental clustering
problems in the incomplete data setting. Namely, we consider the following three problems, referred to as \INCLUSq, \ANYCLUSq, and \PAIRCLUSq, that share a similar setting:
In all three problems, the input is a (multi)set
$M$ of $d$-dimensional vectors over the Boolean domain\footnote{We view $M$ as the (multi)set of rows of a Boolean matrix.}, some of whose entries might be
missing, and two parameters $r, k \in \nn$\sv{; we use the symbol $\blank$ to denote a missing vector entry}.
\sv{For \INCLUSq{}, the goal is decide whether the $\blank$ entries in $M$
can be completed to obtain a (complete) set of vectors $M^*$ such that
there is a subset $S \subseteq M^*$ with $|S|\leq k$ satisfying that,
for every $\vec{a} \in M^*$, the Hamming distance between $\vec{a}$
and some vector in $S$ is at most $r$. That is, the goal for
\INCLUSq{}
is to complete the missing entries so as to enable a partitioning
of the resulting (complete) set into at most $k$ clusters such that all vectors in the
same cluster are within Hamming distance at most $r$ from some ``center''
vector that belongs to the cluster itself. For \ANYCLUSq{}, the goal is
the same as that for \INCLUSq{}, except that the center vectors need
not be in the set $M$ (\ie, are chosen from $\{0, 1\}^d$).  For
\PAIRCLUSq{}, the goal is to complete the missing entries in $M$ so as to obtain a set $M^*$ such that the vectors
of $M^*$ can be partitioned into at most~$k$ clusters/subsets, each of
diameter at most~$r$ (where the diameter of a set of vectors $S$ is
the maximum pairwise Hamming distance over all pairs of vectors in the
set). We denote by \INCLUS, \ANYCLUS, and \PAIRCLUS{}
the complete versions of \INCLUSq, \ANYCLUSq, and \PAIRCLUSq,
respectively; that is, the restrictions of the aforementioned data
completion problems to input instances in which the set of vectors
contains no missing entries. }

\lv{
 For \INCLUSq{}, the goal
is to complete the missing entries so as to enable a partitioning
of the set $M$ into at most $k$ clusters such that all vectors in the
same cluster are within distance at most $r$ from some ``center''
vector that belongs to the cluster itself. The goal for \ANYCLUSq{} is
the same as that for \INCLUSq{}, except that the center vectors need
not be in the set $M$ (\ie, are chosen from $\{0, 1\}^d$). For
\PAIRCLUSq{}, the goal is to complete the missing entries so as to
enable a partitioning of $M$ into at most~$k$ clusters such that the
diameter of each cluster is at most~$r$. The formal problem definitions are given
in Section~\ref{section:prelims}.
}

Our first order of business is to obtain a detailed map of the parameterized complexity of
the above three data completion problems. As we show in this paper, parameterization by $k+r$ is not sufficient to achieve tractability for any of these three problems: one needs to restrict the occurrences
of the unknown entries in some way as well. We do so by adopting a
third parameter defined as the minimum number of vectors and
coordinates (or, equivalently, rows and columns in a matrix representation of~$M$)
needed to cover all the missing entries. This parameter, which we call the \emph{covering number} or simply $\comb$, is guaranteed to be small when the unknown entries arise from the addition of a small number of new rows and columns (\eg, new users and attributes) into a known data-set; in particular, the parameter may be small even in instances with a large number of rows and columns that contain missing entries. The covering number has previously been used in the context of matrix completion~\cite{icml} and is in fact the least restrictive parameter considered in that  paper.

Our main contribution is a complete parameterized complexity landscape
for the complete and incomplete versions of all three clustering
problems w.r.t.\ all combinations of the parameters $k$, $r$, and
$\comb$. Our main algorithmic contribution shows that the
incomplete variants of all three clustering problems
are fixed-parameter tractable parameterized by
$k+r+\comb$, and as a consequence the complete variants are fixed-parameter tractable
parameterized by $k+r$. Notably, our tractability results are obtained
using
\emph{kernelization}~\cite{FominLokshtanovSaurabhZehavi19,GaspersSzeider14}
and therefore provide efficient
polynomial-time preprocessing procedures, which can be applied before
the application of any available (even heuristic) clustering
algorithm. To perform the kernelization, we apply a two-step approach: first we build on the well-known Sunflower Lemma~\cite{Erdos60} to develop new tools that allow us to reduce the number of rows in the target instance, and after that we use entirely different techniques to identify a small set of ``distance-preserving'' relevant coordinates.
Together with a set of algorithmic lower
bound results (and an \XP{} algorithm for \INCLUS{} parameterized by
$k$), this provides the comprehensive parameterized complexity
landscape illustrated in Table~\ref{tab:results-combined}.
We also show that all our tractability results can be lifted from the
Boolean domain to any finite domain, for the Hamming
as well as Manhattan distance.

\begin{table*}[htbp]
  \begin{center}
    \begin{tabular}{@{}c@{\qquad\qquad~}c@{\qquad}c@{\qquad}c@{\qquad}c@{}}\toprule
      Parameter: & $k$ & $r$ & $k+r$ & $k+r+\comb$ \\ \midrule
      \INCLUS & \W{2}\hy c & \paraNP{}\hy c &  \FPT{} & N/A \\
      \CLUS{Any/DIAM} & \paraNP{}\hy c & \paraNP{}\hy c &  \FPT{} & N/A
      \\
\midrule
      \TCLUSq{In/Any/DIAM} & \paraNP{}\hy c
                    & \paraNP{}\hy c & \paraNP{}\hy c & \FPT{} \\
     \bottomrule
    \end{tabular}
  \end{center}

  \caption{Parameterized complexity results for exact clustering with
    complete data (top) and incomplete data (bottom). \FPT{} means
    fixed-parameter tractability, while \paraNP{}\hy c and \W{2}\hy c mean
    completeness for these complexity classes and indicate
    fixed-parameter \underline{in}tractability\lv{(see
    Section~\ref{section:prelims})}.}
  \label{tab:results-combined}
\end{table*}

\smallskip
\noindent \textbf{Related Work.}
In previous work, Hermelin and Rozenberg~(\citeyear{hermelin}) studied the {\sc Closest String with Wildcards} problem, which corresponds to \ANYCLUSq{} with $k=1$. Independently of our work, Koana et al.~(\citeyear{niedermeierradius}) very recently revisited the earlier work of Hermelin and Rozenberg~(\citeyear{hermelin}) and obtained, among other results, a fixed-parameter algorithm for that problem parameterized by $r$ plus the maximum number of missing entries per row; in that same paper, they also studied \INCLUSq\ with $k=1$.
Even more recently, the same group~\cite{niedermeierdiameter} also studied a problem related to \PAIRCLUSq\ for a single cluster, \ie, for $k=1$. They obtain a classification orthogonal to ours w.r.t.\ constant lower and upper bounds on the diameter and the maximum number of missing entries per row.

The main differences between the problems studied by Koana et al.~(\citeyear{niedermeierradius,niedermeierdiameter}) and the restrictions of \ANYCLUSq{} and \INCLUSq{} (studied in this paper) to $k=1$ (\ie, the restriction to the special case where we seek precisely 1 cluster) is the parameter used to capture the number of missing entries per row. Indeed, the authors of these works consider the maximum number of missing entries (over all rows), whereas we consider the parameter $\comb$. The two parameters are orthogonal:
there are instances in which the maximum number of missing entries per row is very small yet $\comb$ is large, and vice versa.

The parameterized complexity of a related problem---\textsc{Matrix Completion}---has been studied in a different context than that of clustering~\cite{icml}; the problem considered therein corresponds to the special case of \INCLUSq{} in which the clustering radius $r$ is $0$.
There is also an extensive body of research on clustering problems for complete data. Examples include the work of Frances and Litman~(\citeyear{litman}), Gramm, Niedermeier and Rossmanith~(\citeyear{GrammNiedermeierRossmanith03}), as well as many other works~\cite{boucherma,cabello,abs-1807-07156,FominGP18,FominGS19,lingashammingcenter,lingasapxclustering,gonzalez}.
Note also that \INCLUS{} and \ANYCLUS{} are special instances of the well-known $k$-center problem.

We remark that related problems have also been studied by a variety of other authors, such as, e.g., Chen, Hermelin, Sorge~(\citeyear{ChenHS19}).

\lv{
\smallskip
\noindent \textbf{Paper Organization.} The paper is structured as
follows. After introducing the relevant preliminaries in
Section~\ref{section:prelims} we introduces the tools that lie at the
core of our approach in
Section~\ref{sec:structural}. Section~\ref{sec:fpt-incomplete} then
employs these tools to obtain \FPT-algorithms for
\IADCLUSq. Section~\ref{sec:hardness} is where we present all the
lower bounds required to obtain the complexity classification
presented in Table~\ref{tab:results-combined}. The final two sections
deal with general implications of our results: Section~\ref{sec:boundeddomain} shows how our results generalize to matrices (vectors) over any bounded domain, while Section~\ref{sec:implications} discusses implications for graph problems. We conclude the paper in Section~\ref{sec:concl} with some remarks and open questions.
}

\lv{
  \section{Preliminaries}\label{section:prelims}

Let $\vec{a}$ and $\vec{b}$ be two vectors in $\{0,1,\blank\}^d$, where $\blank$ is used to represent coordinates whose value is unknown (\ie, missing entries). We denote
by $\HSET(\vec{a},\vec{b})$ the set of coordinates in which
$\vec{a}$ and $\vec{b}$ are guaranteed to differ, \ie,
$\HSET(\vec{a},\vec{b})=\SB i
\SM (\vec{a}[i]=1 \wedge \vec{b}[i]=0)\vee (\vec{a}[i]=0 \wedge \vec{b}[i]=1) \SE$, and we denote by
$\HDIST(\vec{a},\vec{b})$ the \emph{Hamming distance} between
$\vec{a}$ and $\vec{b}$ measured only between known entries, \ie,
$|\HSET(\vec{a},\vec{b})|$. Moreover, for a subset $D' \subseteq [d]$
of coordinates, we denote by $\vec{a}[D']$ the vector $\vec{a}$
restricted to the coordinates in $D'$.

There is a one-to-one correspondence between vectors in
$\{0,1\}^d$ and
subsets of coordinates, \ie, for every vector, we can associate the unique
subset of coordinates containing all its one-coordinates and
vice-versa.
We introduce the following notation for vectors to switch between their set-representation
and vector-representation. We denote by $\HSET(\vec{a})$ the set
$\HSET(\vec{0},\vec{a})$, and for a set $C$ of coordinates, we denote
by $\HVEC(C)$ the vector that is $1$ at precisely the coordinates in $C$ and $0$ at all
other coordinates. We extend this notation to sets of
vectors and a family of coordinate sets, respectively. For a set $N$ of
vectors in $\{0,1\}^d$ and a family ${\cal C}$ of coordinate sets, we denote by $\HSET(N)$ the set $\SB\HSET(\vec{v}) \SM \vec{v} \in N\SE$ and by $\HVEC(\CCC)$ the set $\SB
\HVEC(C) \SM C \in \CCC\SE$. We say that a vector $\vec{a}\in \{0,1\}^d$
is a \emph{$t$-vector} if $|\HSET(\vec{a})|=t$ and
we say that $\vec{a}$ \emph{contains} a subset $S$ of coordinates if $S \subseteq \HSET(\vec{a})$.
For a subset $S \subseteq \{0,1\}^d$ and a vector $\vec{a} \in \{0,1\}^d$,
we denote by $\HDIST(S,\vec{a})$ the minimum Hamming distance between
$\vec{a}$ and the vectors in $S$, \ie,
$\HDIST(S,\vec{a})=\min_{\vec{s} \in
  S}\HDIST(\vec{s},\vec{a})$.
We denote by $\DIAM(S)$ the
\emph{diameter} of $S$, \ie, $\DIAM(S):=\max_{\vec{s_1},\vec{s_2} \in
	S}\HDIST(\vec{s_1},\vec{s_2})$.

Let $M\subseteq \{0,1\}^d$ and let $[d]=\{1,\dots,d\}$. For a vector $\vec{a} \in M$, we
denote by $\HN{\vec{a}}{r}$ the
\emph{$r$-Hamming neighborhood} of $\vec{a}$, \ie, the set $\SB
\vec{b} \in M \SM \HDIST(\vec{a},\vec{b})\leq r\SE$ and by $\HN{M}{r}$ the
set $\bigcup_{\vec{a} \in M}\HN{\vec{a}}{r}$.
Similarly, we denote by $\HN{\vec{a}}{=r}$ the
the set $\SB \vec{b} \in M \SM \HDIST(\vec{a},\vec{b}) = r\SE$ and by $\HN{M}{=r}$ the
set $\bigcup_{\vec{a} \in M}\HN{\vec{a}}{=r}$.
We say that $M^*\subseteq \{0,1\}^d$ is a \emph{completion} of $M\subseteq \{0,1,\blank\}^d$ if there is a bijection $\alpha: M \rightarrow M^*$ such that for all $\vec{a}\in M$ and all $i\in [d]$ it holds that either $\vec{a}[i]=\blank$ or $\alpha(\vec{a})[i]=\vec{a}[i]$.

We now proceed to give the formal definitions of the problems under consideration.

\pbDef{\INCLUSq{}}{A subset $M$ of $\{0,1,\blank\}^d$ and $k, r \in \Nat$.}{Is there a completion $M^*$ of $M$ and a subset $S \subseteq M^*$ with $|S|\leq k$ such that $\HDIST(S,\vec{a})\leq r$ for every $\vec{a} \in M^*$?}
 
\pbDef{\ANYCLUSq{}}{A subset $M$ of $\{0,1,\blank\}^d$ and $k, r \in \Nat$.}{Is there a completion $M^*$ of $M$ and a subset $S \subseteq \{0,1\}^d$ with $|S|\leq k$ such that $\HDIST(S,\vec{a})\leq r$ for every $\vec{a} \in M^*$?}
 
\pbDef{\PAIRCLUSq{}}{A subset $M$ of $\{0,1,\blank\}^d$ and $k, r \in \Nat$.}{Is there a completion $M^*$ of $M$ and a partition $\PPP$ of $M^*$ with $|\PPP|\leq k$
	such that $\DIAM(P)\leq r$ for every $P \in \PPP$?}

Observe that in a matrix representation of the above problems, we can represent the input matrix as a \emph{set} of vectors where each row of the matrix corresponds to one element in our set. Of course, to precisely capture the input it seems more appropriate to consider multisets of vectors---however this
is not an issue here, since even if the initial matrix contained multiple copies of a row, removing any such row from the instance will not change the outcome for any of the three clustering problems considered above. Note that this is not the case for the remaining problems considered in this paper: there it will be important to keep track of repeated rows, and so we will correctly treat $M$ as a multiset. In the few cases where we perform a union of two multisets, we will assume it to be a disjoint union.


We remark that even though the statements are given in the form of
decision problems, all tractability results presented in this paper are constructive and the associated algorithms can also output a solution (when it exists) as a witness, along with the decision.
In the case where we restrict the input to vectors over $\{0,1\}^d$ (\ie, where all entries are known), we omit ``-\textsc{Completion}'' from the problem name.

In some of the proofs, it will sometimes be useful to argue using the \emph{compatibility graph} associated with an instance $\inst$.
This graph, denoted by $\compG(\inst)$, is the undirected
graph that has a vertex for every vector in $M$, and an edge between two
vectors $\vec{a}$ and $\vec{b}$ if and only if:
\begin{itemize}
	\item $\HDISTq(\vec{a},\vec{b})\leq r$ (if $\inst$ is an instance of
	\INCLUS{} or \PAIRCLUS{}), or
	\item $\HDISTq(\vec{a},\vec{b})\leq 2r$ (if $\inst$ is an instance of
	\ANYCLUS{}).
\end{itemize}

We observe that vectors in different connected components of $\compG$ cannot interact
with each other at all: every cluster containing vectors from one
connected component cannot contain a vector from any other connected
component.

\smallskip
\noindent \textbf{Parameterized Complexity.}
In parameterized complexity~\cite{FlumGrohe06,DowneyFellows13,CyganFKLMPPS15},
the complexity of a problem is studied not only with respect to the input size, but also with respect to some problem parameter(s). The core idea behind parameterized complexity is that the combinatorial explosion resulting from the \NP-hardness of a problem can sometimes be confined to certain structural parameters that are small in practical settings. We now proceed to the formal definitions.

A {\it parameterized problem} $Q$ is a subset of $\Omega^* \times
\mathbb{N}$, where $\Omega$ is a fixed alphabet. Each instance of $Q$ is a pair $(I, \kappa)$, where $\kappa \in \Nat$ is called the {\it
parameter}. A parameterized problem $Q$ is
{\it fixed-parameter tractable} (\FPT)~\cite{FlumGrohe06,DowneyFellows13,CyganFKLMPPS15}, if there is an
algorithm, called an {\em \FPT-algorithm},  that decides whether an input $(I, \kappa)$
is a member of $Q$ in time $f(\kappa) \cdot |I|^{\bigoh(1)}$, where $f$ is a computable function and $|I|$ is the input instance size.  The class \FPT{} denotes the class of all fixed-parameter
tractable parameterized problems.

A parameterized problem $Q$
is {\it \FPT-reducible} to a parameterized problem $Q'$ if there is
an algorithm, called an \emph{\FPT-reduction}, that transforms each instance $(I, \kappa)$ of $Q$
into an instance $(I', \kappa')$ of
$Q'$ in time $f(\kappa)\cdot |I|^{\bigoh(1)}$, such that $\kappa' \leq g(\kappa)$ and $(I, \kappa) \in Q$ if and
only if $(I', \kappa') \in Q'$, where $f$ and $g$ are computable
functions. By \emph{\FPT-time}, we denote time of the form $f(\kappa)\cdot |I|^{\bigoh(1)}$, where $f$ is a computable function.
Based on the notion of \FPT-reducibility, a hierarchy of
parameterized complexity, {\it the \cc{W}-hierarchy} $=\bigcup_{t
\geq 0} \W{t}$, where $\W{t} \subseteq \W{t+1}$ for all $t \geq 0$, has
been introduced, in which the $0$-th level \W{0} is the class {\it
\FPT}. The notions of hardness and completeness have been defined for each level
\W{$i$} of the \cc{W}-hierarchy for $i \geq 1$. It is commonly believed that $\W{1} \neq \FPT$, and the
\W{1}-hardness has served as the main working hypothesis of fixed-parameter
intractability. The class \XP{} contains parameterized problems that can be solved in time  $\bigoh(|I|^{f(\kappa)})$, where $f$ is a computable function; it
contains the class \W{t}, for $t \geq 0$, and every problem in \XP{} is polynomial-time solvable when the parameters are bounded by a constant.
The class \paraNP{} is the class of parameterized problems that can be solved by non-deterministic algorithms in time $f(\kappa)\cdot |I|^{\bigoh(1)}$, where $f$ is a computable function.
A problem is \emph{\paraNP{}-hard} if it is \NP-hard for a constant value of the parameter~\cite{FlumGrohe06}.

A parameterized problem is {\em kernelizable}
if there exists a polynomial-time reduction that maps an instance $(I, \kappa)$ of
the problem to another instance $(I', \kappa')$ such that (1) $|I'| \leq f(\kappa)$ and $\kappa' \leq f(\kappa)$, where $f$ is a computable function, and (2) $(I,\kappa)$ is a \yes-instance
of the problem if and only if $(I',\kappa')$ is. The instance
$(I',\kappa')$ is called the {\em kernel} of~$I$. It is well known that a
decidable problem is \FPT{} if and only if it is
kernelizable~\cite{DowneyFellows13}.
A \emph{polynomial kernel} is a kernel whose size can be bounded by a
polynomial in the parameter.


\smallskip
\noindent \textbf{Structure of Missing Entries.}
As we will later show, it is not possible to obtain fixed-parameter tractability for clustering or finding a large cluster
when the occurrence of missing entries (\ie, $\blank$'s) is not restricted in any way. On the other hand, when we restrict the total number of missing entries to be upper bounded by a constant or a parameter, our problems trivially reduce to the complete-data setting,
since one can enumerate all values in these few missing entries by brute force.
Hence, the interesting question is whether we can solve the
problems when the number of $\blank$'s is large but also restricted
in a natural way. We do so by using the so-called \emph{covering number} as a parameter, a setting which naturally captures instances where incomplete data is caused by the addition of a few new vectors (\ie, rows) and/or coordinates (\ie, columns)~\cite{icml}.

Formally, let $\{\vec{v}_1,\ldots \vec{v}_n\}$ be an arbitrary but fixed ordering of a subset $M$ of $\{0,1,\blank\}^d$. If $\vec{v}_i[j]=\blank$, we say that $\blank$ at $\vec{v}_i[j]$ is \emph{covered} by row $i$ and column $j$. The \emph{covering number} of $M$, denoted as $\comb(M)$ or simply $\comb$ where it is clear from the context, is the minimum value of $r+c$ such that there exist $r$ rows and $c$ columns in $M$ with the property that each occurrence of $\blank$ is covered by one of these rows or columns. We will generally assume that for a set $M\in \{0,1,\blank\}^d$ we have computed sets $T_M$ and $R_M$ such that $\comb(M)=|T_M|+|R_M|$ and each $\blank$ occurring in a vector $\vec{v}\in M$ is covered by a row in $R_M$ or a column in $T_M$; we note that this computation may be done in polynomial time~\cite[Proposition 1]{icml}, and in our algorithms parameterized by $\comb(M)$ we will generally assume that $T_M$ and $R_M$ have already been pre-computed.

}

\section{The Toolkit}\label{sec:structural}
In this section, we present key structural results that are employed in several algorithms and lower bounds in the paper. The first part of our toolkit and structural results for matrices are obtained by exploiting the classical sunflower lemma of Erd\"os and Rado, a powerful combinatorial tool that has been used to obtain kernelization algorithms for many fundamental parameterized problems~\cite{FominLokshtanovSaurabhZehavi19}.
A \emph{sunflower} in a set family $\FFF$ is a subset $\FFF' \subseteq \FFF$ such that all pairs of elements in $\FFF'$ have the same intersection.

%

\begin{lemma}[Erd\"os and Rado~\citeyear{Erdos60}; Flum and Grohe~\citeyear{FlumGrohe06}]\label{lem:SF}
  Let $\FFF$ be a family of subsets of a universe $U$, each of cardinality exactly
  $b$, and let $a \in \mathbb{N}$. If $|\FFF|\geq b!(a-1)^{b}$, then $\FFF$
  contains a sunflower $\FFF'$ of cardinality at least $a$. Moreover,
  $\FFF'$ can be computed in time polynomial in $|\FFF|$.
\end{lemma}

\myparagraph{Finding Irrelevant Vectors.}
\label{subsec:sunflower}
The first structural lemma we introduce is Lemma~\ref{lem:CLUS-SF-BASIC}, which is also illustrated in Figure~\ref{fig:sunflower}. Intuitively speaking, the lemma says that if the $t$-Hamming neighborhood of a vector $\vec{v}$ contains a large sunflower, then at least one of its elements can be
removed without changing the maximum distance to any vector $\vec{a}$ that is of
distance at most $r$ to the elements in the sunflower. The proof of Lemma~\ref{lem:CLUS-SF-BASIC} utilizes the straightforward Lemma~\ref{lem:CLUS-SF-BASIC-A}, which captures a useful observation that is also used in other proofs.

We note that the idea of applying the Sunflower Lemma on a similar set representation of an instance was also used in a previous work by Marx~(\citeyear{marx}) (see also Kratsch, Marx and Wahlstr\"om, \citeyear{kmw}) to obtain \FPT{} and kernelization results, albeit in the context of studying the weighted satisfiability of CSPs.
There, the authors used the sunflower to reduce the arity of constraints by replacing the sets in the sunflower (which correspond to the scope of the constraints) by constraints defined over the petals without the core plus one additional constraint defined only on the variables of the core. We, however, use the sunflower in a different manner, namely to identify irrelevant vectors that can be safely removed from the instance.
Note also that in contrast to many other applications of the sunflower, where all petals are removed and replaced by the core, this is not possible in our setting since we need to keep a certain number of petals in order to maintain the clustering properties of the instance.
%
%

\begin{figure}

  \begin{minipage}[c]{3.8cm}
      \begin{equation*}
        \begin{array}{cc}
          \vec{v}= (0, 0, 0, 0, 0, 0, 0, 0)\\[5mm]
          \vec{a} = (1, 1, 1, 1, 1, 0, 0, 0)\\
        \end{array}
      \end{equation*}
    \end{minipage}
    \begin{minipage}[c]{3.5cm}
      \includegraphics[width=3.4cm]{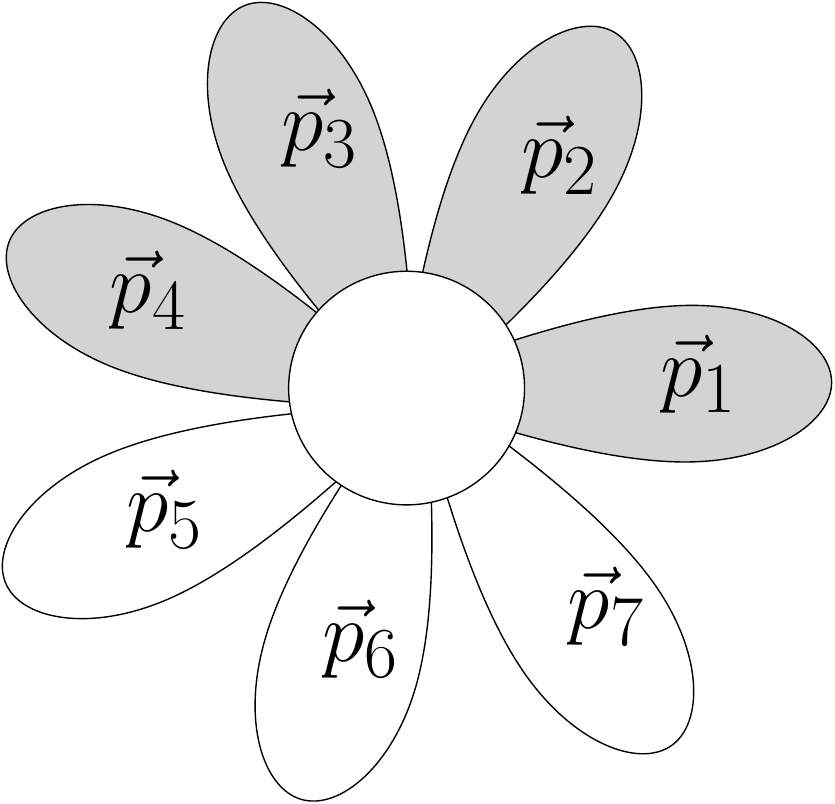}
    \end{minipage}
    \caption{The figure shows an example for the setting in
      Lemma~\ref{lem:CLUS-SF-BASIC}.  Here
      $r=3$ and $t=2$ and the figure shows the vectors $\vec{v}$ and $\vec{a}$ as well as the
      sunflower resulting from the vectors $\vec{p}_1,\dotsc,\vec{p}_7$
      with $\vec{p}_i[j]=1$ if and only if either $j=1$ or $j=i+1$.
      In this example three of the petals, \ie, the
      white petals $\vec{p}_5$, $\vec{p}_6$, and $\vec{p}_7$, only share
      the core of the sunflower with $\vec{a}$, which implies that all
      three of these petals are of maximum Hamming distance to $\vec{a}$. 
    }  \label{fig:sunflower}
\end{figure}

\lv{\begin{lemma}}
  \sv{\begin{lemma}}
    \label{lem:CLUS-SF-BASIC-A}
  Let $t, r \in \Nat$. Let $N \subseteq \{0,1\}^d$ be a set
  of $t$-vectors such that $\FFF:=\HSET(N)$ is a sunflower with core
  $C$. If $|N|>r$, then
  for every vector $\vec{a} \in \{0,1\}^d$ with $|\HSET(\vec{a})|\leq
  r$, $N$ contains a vector that has maximum distance to $\vec{a}$ among all
  $t$-vectors that contain $C$.
\end{lemma}
\lv{\begin{proof}
  Let $\vec{a} \in \{0,1\}^d$ with $|\HSET(N,\vec{a})|\leq r$ be
  arbitrary. Since $|N|>r$, there is
  a vector $\vec{n}\in N$ such that $\HSET(\vec{n})\cap
  \HSET(\vec{a})\subseteq C$, which implies that the distance of
  $\vec{n}$ to $\vec{a}$ is maximum among all $t$-vectors that contain $C$.
\end{proof}
}

\lv{\begin{lemma}}
  \sv{\begin{lemma}}
    \label{lem:CLUS-SF-BASIC}
  Let $t, r \in \Nat$, $\vec{v} \in \{0,1\}^d$, $N
  \subseteq \HN{\vec{v}}{=t}$, and $\FFF:=\SB \HSET(\vec{v},\vec{x}) \SM x
  \in N \SE$ such that $\FFF$ is a sunflower with core $C$.
  Then for every subset $N'$ of $N$ with $|N'|\geq r+t+2$ and
  every vector $\vec{a} \in \{0,1\}^d$ such that $\HDIST(N',\vec{a})\leq r$,
  we have
  $\HDIST(\vec{f},\vec{a})\leq \max_{\vec{x} \in N' \setminus
    \{\vec{f}\}}\HDIST(\vec{x},\vec{a})$
  for every $\vec{f} \in N$.
\end{lemma}
\lv{\begin{proof}

    Figure~1 
    illustrates an example situation for this
  lemma.
  Since $\HDIST(N',\vec{a})\leq r$ and $\HDIST(\vec{x},\vec{v})=t$
  for every $\vec{x}\in N$, we obtain that
  $\HDIST(\vec{v},\vec{a})\leq r+t$. Since $|N'\setminus
  \{\vec{f}\}|> r+t$, we obtain from Lemma~\ref{lem:CLUS-SF-BASIC-A}
  that $N'\setminus \{\vec{f}\}$ contains a vector that has maximum
  distance to $\vec{a}$ among all $t$-vectors that contain $C$, and
  hence in particular to all vectors in $N$.
\end{proof}
}
The following lemma now employs Lemmas~\ref{lem:CLUS-SF-BASIC} and~\ref{lem:SF} to show that if the
$t$-Hamming neighborhood of a vector $\vec{v}$ is large enough, at least
one of its elements can be removed without changing the clustering
properties of the instance.
\lv{\begin{lemma}}
  \sv{\begin{lemma}}
    \label{lem:CLUS-SF-NEIGH}
    Let $k, r, t \in \Nat$, $M \subseteq \{0,1\}^d$,
    $\vec{v} \in M$, and $N:=\HN{\vec{v}}{=t}\cap M$. If $|N|\geq
    t!(k(r+t+2))^{t}$, then there is a vector $\vec{f} \in N$
    satisfying the following two properties:
    \lv{  \begin{enumerate}[leftmargin=*,widest=(P1)]
      \item[(P1)] for every set $S \subseteq \{0,1\}^d$ with $|S|\leq k$ and satisfying
        $\HDIST(S,\vec{m})\leq r$ for every $\vec{m} \in M$,
        it holds that $\max_{\vec{y} \in M}\HDIST(S,\vec{y})=\max_{\vec{y} \in M\setminus \{\vec{f}\}}\HDIST(S,\vec{y})$; and
      \item[(P2)] $M$ has a partition
        into at most $k$ clusters, each of diameter at most $r$, if and only if
        $M\setminus \{\vec{f}\}$ does.
      \end{enumerate}
    }
    \sv{
      \textbf{\textup{(P1)}} for every set $S \subseteq \{0,1\}^d$ with $|S|\leq k$ and satisfying
      $\HDIST(S,\vec{m})\leq r$ for every $\vec{m} \in M$,
      it holds that $\max_{\vec{y} \in M}\HDIST(S,\vec{y})=\max_{\vec{y} \in M\setminus \{\vec{f}\}}\HDIST(S,\vec{y})$; and
      \textbf{\textup{(P2)}} $M$ has a partition
      into at most $k$ clusters, each of diameter at most $r$, if and only if
      $M\setminus \{\vec{f}\}$ does.
    }Moreover, $\vec{f}$ can be determined in time polynomial in $M$.
  \end{lemma}

  \lv{\begin{proof}
      \sv{\begin{proof}[Proof Sketch]}
        \newcommand{\HCL}{\Gamma}
        Let $\FFF:=\SB \HSET(\vec{v},\vec{x}) \SM \vec{x} \in N \SE$. Then
        $|\FFF|=|N|\geq t!(k(r+t+2))^{t}$ and $|F|=t$ for every $F \in
        \FFF$. By Lemma~\ref{lem:SF}, $\FFF$ contains a
        sunflower, say $\FFF'$, of size larger than $k(r+t+2)$ (with core $C$).

        We denote by $N(F)$
        the vector in $N$ giving rise to the element $F \in \FFF$, \ie,
        $F=\HSET(\vec{v},N(F))$. Moreover, for a subset $\FFF''$ of
        $\FFF$ we denote by $N(\FFF'')$ the set $\SB N(F) \SM F \in \FFF''\SE$.
        Let $F \in \FFF'$ be arbitrarily chosen. We claim that setting $\vec{f}$ to the vector
        $N(F)$ satisfies the claim of the lemma.
        Since $|\FFF'\setminus \{F\}|\geq k(r+t+2)$, we obtain
        that:
        \begin{enumerate}[leftmargin=*,widest=(1)]
        \item[(1)] for every set $S \subseteq \{0,1\}^d$ with $|S|\leq k$ there
          is a vector $\vec{s} \in S$ such that $|N'|\geq r+t+2$, where $N':=\SB \vec{u} \in
          \{0,1\}^d \SM \HDIST(\vec{s},\vec{u})= \HDIST(S,\vec{u}) \SE \cap N(\FFF')$.
        \item[(2)] For every partition $\PPP$ of $M\setminus \{\vec{f}\}$ into at most $k$ sets
          there is a set $P \in \PPP$ with $|P \cap N(\FFF'\setminus
          \{F\})|\geq r+t+2$.
        \end{enumerate}

        Towards showing (P1), let $S\subseteq \{0,1\}^d$ with $|S|\leq
        k$. By (1), there is a vector $\vec{s} \in S$ with $|N'|\geq
        r+t+2$. Then $t$, $r$,
        $\vec{v}$, $N$, and $N'$ satisfy the conditions of
        Lemma~\ref{lem:CLUS-SF-BASIC}. By observing that
        $\HDIST(N',\vec{s})\leq r$, we obtain that
        $\HDIST(\vec{f},\vec{s}) \leq \max_{x \in N'\setminus
          \{\vec{f}\}}\HDIST(\vec{x},\vec{s})$. Consequently, $\max_{\vec{y}
          \in M}\HDIST(S,\vec{y})=\max_{\vec{y} \in
          M\setminus\{\vec{f}\}}\HDIST(S,\vec{y})$, which shows (P1).
        \sv{To establish (P2), we note that the forward direction is trivial and for
          the backward direction we use (2) and Lemma~\ref{lem:CLUS-SF-BASIC}.
        }

        Towards showing (P2), first note that the forward direction holds
        trivially. Towards showing the other direction, let $\PPP$ be any
        partition of $M \setminus \{\vec{f}\}$ into at most $k$ sets, each of
        diameter at most $r$. By (2), there is a set $P \in \PPP$
        with $|P\cap N(\FFF'\setminus \{F\})|\geq r+t+2$. Let $N'$ be the
        set $(P\cap N(\FFF'\setminus \{F\}) \cup \{\vec{f}\}$.
        Then $t$, $r$,
        $\vec{v}$, $N$, $N'$
        satisfy the conditions of Lemma~\ref{lem:CLUS-SF-BASIC}. By
        observing that $\HDIST(N',\vec{p})\leq r$ for every $\vec{p}\in P$,
        we obtain that $\HDIST(\vec{f},\vec{p}) \leq \max_{x \in N'\setminus
          \{\vec{f}\}}\HDIST(\vec{x},\vec{p})$ for every $\vec{p} \in
        P$. Hence $P\cup \{\vec{f}\}$ has diameter at most $r$, which
        implies that the partition obtained from $\PPP$ after adding $\vec{f}$
        to $P$ is a partition of $M$ into at most $k$ clusters, each of diameter
        at most $r$.
      \end{proof}}

\myparagraph{Irrelevant Coordinates and Diameter Bound.}
\label{ssec:irrelevantcoors}
Our clustering algorithms for \IADCLUSq{} will broadly proceed in two
steps. Given an instance $\inst=(M,k,r)$ of \IADCLUSq{}, we will first
compute an equivalent instance $(M',k,r)$ such that the size of $M'$
can be bounded by a function of the parameter $k+r+\comb(M)$ (this is done by the \emph{irrelevant vector technique}). However,
since our aim is to obtain a kernel, we then still need to reduce the
number of coordinates for every vector in $M'$. That is where we use our \emph{irrelevant coordinate technique}. This subsection introduces the tools and notions that are central to this technique. Throughout this section, we will assume that $\inst=(M,k,r)$ is
the considered input instance of \IADCLUSq{}.

Let $\DCOOR(M)$ for $M \subseteq \{0,1\}^d$ be the set of
all coordinates $i$ such that at least two vectors in $M$ disagree on
their $i$-th coordinate, \ie, there are two vectors $\vec{y},\vec{y}'\in M$
such that $\{\vec{y}[i],\vec{y}'[i]\}=\{0,1\}$.
Intuitively,
$\DCOOR(M)$ is the set of \emph{important coordinates}, since all other
coordinates can be safely removed from the instance; this is because
they can always be completed to the same value and hence do not
influence the properties of a clustering of $M$. Note that if we could
show that the number of important coordinates is bounded by a function
of $M'$ and our parameter $k+r+\comb(M)$, then we would obtain a
kernel by simply removing all coordinates that are not
important. Unfortunately, this is not the case for two reasons: First
the compatibility graph $\compG(\III)$ can consist of more than one
component and the vectors in different components can differ in
arbitrary coordinates. Furthermore, even inside a component the number
of important coordinates can be arbitrary large. For instance, a
component could consist of the all-zero vector, the all-one vector,
and the all $\blank$ vector. Note that the all $\blank$ vector is
crucial for this example and indeed, the next lemma shows that if we restrict
ourselves to a component containing only vectors in $M\setminus R_M$,
then the number of important coordinates can be bounded in terms of
the diameter and the number of vectors inside the component.

\lv{\begin{lemma}}
  \sv{\begin{lemma}}
\label{lem:diam-comp}
  Let $M' \subseteq M\setminus R_M$ such that $\compG(\inst)[M']$ is
  connected. Then $|\DCOOR(M')\setminus T_M| \leq \DIAMq(M')(|M'|-1)$.
\end{lemma}
\lv{\begin{proof}
  Let $\vec{m} \in M'$ be arbitrary. Then for every vector $\vec{m}'
  \in M'$, there are at most $\DIAMq(M')$ coordinates in $[d]\setminus
  T_M$ such that $\vec{m}[i]\neq \vec{m}'[i]$. Therefore in total there are at most
  $(\DIAMq(M'))(|M'|-1)$ coordinates in  $[d]\setminus
  T_M$ for which any vector in
  $M'$ differs from $\vec{m}$ and hence $|\DCOOR(M')\setminus T_M| \leq
  \DIAMq(M')(|M'|-1)$.
\end{proof}}
The next lemma now shows how to bound the diameter of every component
in $M\setminus R_M$ in terms of our parameter $k+r+\comb(M)$.
\lv{\begin{lemma}}
  \sv{\begin{lemma}}
    \label{lem:IN-Q-DIAM-BOUND}
  Let $\inst=(M,k,r)$ be an instance of \INCLUSq{}, \ANYCLUSq{}, or
  \PAIRCLUSq{} and let $M' \subseteq M\setminus R_M$ be such that
  $\compG(\inst)[M']$ is connected. Then $\inst$ is a \NO\hy instance
  if either:
  \begin{itemize}
  \item $\inst$ is an instance of \INCLUSq{} and $\DIAMq(M')>3rk-r+|T_M|$;
  \item $\inst$ is an instance of \ANYCLUSq{} and $\DIAMq(M')>4rk-r+|T_M|$; or
  \item $\inst$ is an instance of \PAIRCLUSq{} and $\DIAMq(M')>2rk-r+|T_M|$.
  \end{itemize}
\end{lemma}
\lv{\begin{proof}
  We start by showing the statement for the case of \INCLUS{}. Assume
  for a contradiction that $\inst$ is a \YES\hy instance of \INCLUS{}
  and let $C_1, \ldots, C_k$ be a partitioning of $M$ into $k$ clusters, each of
  radius at most $r$. Consider any two vectors $\vec{a}, \vec{b} \in
  M'$; let $v_a$ and $v_b$ be their corresponding vertices in
  $\compG(\inst)$, and let $P_{ab}$ be a shortest path between $v_a$ and
  $v_b$ in $\compG(\inst)[M']$, which exists since $\compG(\inst)[M']$ is
  connected. By minimality of $P_{ab}$, $P_{ab}$ cannot
  contain more than three vertices corresponding to vectors in the
  same cluster, as otherwise, $P_{ab}$ could be shortcut by going
  through the center of that cluster. It follows that the length of
  $P_{ab}$ is at most $3k-1$.
  Since every edge in $\compG(\inst)$ represents a Hamming distance of at
  most $r$ between the two vectors to the endpoints of
  the edge and all $\blank$-entries of every vector in $M'$ is
  contained in $T_R$, the lemma follows.
  The proofs for \ANYCLUS{} and \PAIRCLUS{} are analogous.
\end{proof}
}

We now already know how to bound the number of important coordinates
inside a component of $M\setminus R_M$. Unfortunately, as we have
illustrated previously, it is not possible to do the same for $M\setminus R_M$, let alone for the complete vector set $M$. However, the following lemma shows that there is a (small) set
$D'$ of coordinates that satisfy a slightly weaker property: it preserves distances up to $r$ within components of $M\setminus R_M$ as
well as to and between the vectors in $R_M$.
\lv{\begin{lemma}}
  \sv{\begin{lemma}}
    \label{lem:diam-complete}
  Let $M' \subseteq M$ and $r'$ be a natural number. Then there is a subset $D' \subseteq [d]$ of
  coordinates such that:
  \begin{itemize}
  \item(C1) $|D'|\leq (k\MDIAMq(M')+|R_M|(|M'|-1))(r'+1)+|T_M|$ and
  \item(C2) for any two vectors $\vec{m}$ and $\vec{m}'$ in $M'$
    such that $\vec{m}$ and $\vec{m}'$ are in the same component of
    $\compG(\inst)[M']$ or one of $\vec{m}$ or $\vec{m}'$ is in $R_M$,
    it holds
    that $\HDIST(\vec{m}, \vec{m}')=\HDIST(\vec{m}[D'], \vec{m}'[D'])$
    if $\HDIST(\vec{m}, \vec{m}')\leq r'$ and $\HDIST(\vec{m}[D'],
    \vec{m}'[D'])>r'$, otherwise.
  \end{itemize}
  Here $\MDIAMq(M')$ is equal to the maximum
  diameter of any connected component of $\compG(\inst)[M']$.
\end{lemma}
\lv{\begin{proof}
  Note that we can assume w.l.o.g. that $\compG(\III)[M']$ has at most
  $k$ components, since otherwise $\inst$ is a trivial \NO\hy
  instance. But then, we obtain from Lemma~\ref{lem:diam-comp} that the set
  \[
     D_0=\bigcup_{C\textup{ is a component of
       }\compG(\III)[M']}\DCOOR(C)
   \]
  has size at most $k\MDIAMq(M')(|M'|-1)+|T_M|$.
  Moreover, $\HDIST(\vec{m}, \vec{m}')=\HDIST(\vec{m}[D_0],
  \vec{m}'[D_0])$ for any two vectors $\vec{m}$ and $\vec{m}'$ in
  $M'\setminus R_M$ that are in the same component of
  $\compG(\inst)[M']$.
  Hence it only remains to ensure that condition
  (C2) is satisfied if (at least) one of $\vec{m}$ and $\vec{m}'$ is
  in $R_M$. To achieve this we add the following coordinates to $D_0$
  for every two vectors $\vec{m}\in R_M$ and $\vec{m}'\in M'\setminus
  \{\vec{m}\}$:
  \begin{itemize}
  \item if $\HDIST(\vec{m},\vec{m}')\leq r'$, then we add the (at most
    $r'$) coordinates in $\HSET(\vec{m},\vec{m}')$ to $D_0$, otherwise
  \item we add an arbitrary subset of $\HSET(\vec{m},\vec{m}')$ of size
    exactly $r'+1$ to $D_0$.
  \end{itemize}
  Let $D'$ be the set obtained from $D_0$ in this manner.
  Then $D'$ clearly satisfies (C2). Finally,
  $|D'|\leq |D_0|+|R_M|(|M'|-1)(r'+1)$ since we add at most $r'+1$
  coordinates to $D_0$ for every $\vec{m} \in R_M$ and $\vec{m}' \in
  M'\setminus \{\vec{m}\}$.
\end{proof}
}

The following lemma now shows that keeping only the set $D'$ of
coordinates is sufficient to
preserve the equivalence for our three clustering problems.
\lv{\begin{lemma}}
  \sv{\begin{lemma}}
    \label{lem:impot-coor-equiv}\sloppypar
  Let $M' \subseteq M$. Then we can compute a set $D' \subseteq [d]$ of
  coordinates in polynomial-time such that:
  \begin{itemize}
  \item $|D'|\leq (k\MDIAMq(M')+|R_M|(|M'|-1))(2r+1)+|T_M|$ and $(M',k,r)$ is a \YES\hy instance of
    $\ANYCLUSq{}$ if and only if $(M_{D'}',k,r)$ is.
  \item $|D'|\leq (k\MDIAMq(M')+|R_M|(|M'|-1))(r+1)+|T_M|$ and for $X\in \{\probfont{In,Diam}\}$:  $(M',k,r)$ is a \YES\hy instance of \CLUSq{$X$}
  if and only if $(M_{D'}',k,r)$ is.

  \end{itemize}
  Here, $M_{D'}'$ is the matrix obtained from $M'$ after removing all
  coordinates (columns) that are not in $D'$.
\end{lemma}
\lv{\begin{proof}
  We start by showing the result for $\ANYCLUSq{}$.
  Let $D'$ be the set of coordinates obtained from
  Lemma~\ref{lem:diam-complete} for $M'$ and $r'=2r$ satisfying (C1)
  and (C2). Because of (C1), it holds that $|D'|\leq
  (k\DIAMq(M')+|R_M|(|M'|-1))(2r+1)+|T_M|$. It remains to show that
  $(M',k,r)$ and $(M_{D'}',k,r)$ are equivalent instances of $\ANYCLUSq{}$.

  Clearly any solution (\ie, a completion and $k$-clustering for that completion) for $(M',k,r)$ is also a solution for
  $(M_{D'}',k,r)$; since all we did was remove a set of coordinates, all distances in the completion can only become smaller. For the forward direction, let $M_{D'}''$
  be a completion of $M_{D'}'$ leading to a solution with at most $k$ centers $S\subseteq \{0,1\}^{|D'|}$ for $(M_{D'}'',k,r)$, and assume that $S$ is inclusion-minimal.
  Consider a cluster given by a center $\vec{s} \in S$ and let
  $\vec{m}\in M_{D'}''$ be a vector in that cluster, \ie,
  $\HDIST(\vec{s},\vec{m})\leq r$; note that $\vec{m}$ exists since $S$ is
  inclusion-wise minimal. Let $\vec{m}' \in M_{D'}'$ be any other
  vector in that cluster; if the cluster consists only of the vector
  $\vec{m}$, we can replace it with a cluster with center $\vec{m}$
  for the instance $(M',k,r)$. Then $\HDIST(\vec{m},\vec{m}')\leq
  2r$ and because $D'$ satisfies (C2) with $r'=2r$ also $\HDIST(\vec{m}_+,\vec{m}_+')\leq
  2r$, where $\vec{m}_+$ and $\vec{m}_+'$ are the vectors in $M'$
  corresponding to $\vec{m}$ and $\vec{m}'$,
  respectively. Therefore, all coordinates where $\vec{m}_+$ and
  $\vec{m}_+'$ differ are contained in $D'$, which implies that both
  vectors can be completed to be equal in all other coordinates,
  \ie, the coordinates in $[d]\setminus
  \HSET(\vec{m}_+,\vec{m}_+')$; let $\vec{m}_+^c$ be one such
  completion of $\vec{m}_+$. Since we choose $\vec{m}_+'$
  arbitrarily this is true also for all other vectors in the
  cluster. Hence, the vector $\vec{s}' \in \{0,1,\blank\}^d$, which
  is equal to $\vec{s}$ for all coordinates in $D'$ and equal to
  $\vec{m}_+^c$ on all other coordinates can be used as
  a center for $(M',k,r)$ replacing $\vec{s}$. Applying the same
  procedure for all centers in $S$, we obtain a solution for
  $(M',k,r)$, as required.

  We now show the result for $\INCLUSq{}$ and $\PAIRCLUSq{}$.
  Let $D'$ be the set of coordinates obtained from
  Lemma~\ref{lem:diam-complete} for $M'$ and $r'=r$ satisfying (C1)
  and (C2). Because of (C1), it holds that $|D'|\leq
  (k\DIAMq(M')+|R_M|(|M'|-1))(r+1)+|T_M|$. It remains to show that
  $(M',k,r)$ and $(M_{D'}',k,r)$ are equivalent instances of
  $\INCLUSq{}$ respectively $\PAIRCLUS{}$.

  Note that the forward direction of the claim is again trivial for
  both $\INCLUSq{}$ and $\PAIRCLUSq{}$, since we are only removing
  coordinates and hence the distances can only become smaller. Towards
  showing the backward direction, we will distinguish between the case
  for $\INCLUSq{}$ and $\PAIRCLUSq{}$. In the former case, let $M_{D'}''$
  be a completion of $M_{D'}'$ leading to the (inclusion-wise minimal)
  solution $S\subseteq M_{D'}''$ for $(M_{D'}'',k,r)$, \ie, a set of at most $k$ centers, and
  consider a cluster given by a center $\vec{s} \in S$. If the cluster
  does not contain any other vector apart from $\vec{s}$, then we can
  replace $\vec{s}$ with the corresponding vector in $M'$ (and use any
  completion). Otherwise, let $\vec{m} \in M_{D'}''$ be a vector in
  the cluster distinct from $\vec{s}$, \ie,
  $\HDIST(\vec{s},\vec{m})\leq r$. Because $D'$ satisfies (C2) with
  $r'=r$ also $\HDIST(\vec{s}_+,\vec{m}_+)\leq
  r$, where $\vec{s}_+$ and $\vec{m}_+$ are the vectors in $M'$
  corresponding to $\vec{s}$ and $\vec{m}$,
  respectively. Therefore, all coordinates where $\vec{s}_+$ and
  $\vec{m}_+$ differ are contained in $D'$, which implies that both
  vectors can be completed to be equal in all other coordinates,
  \ie, the coordinates in $[d]\setminus
  \HSET(\vec{s}_+,\vec{m}_+')$; let $\vec{s}_+^c$ be one such
  completion for $\vec{s}_+$. Since we choose $\vec{m}_+$
  arbitrarily this is true also for all other vectors in the
  cluster. Hence, we can use the vector $\vec{s}_+^c$ as a replacement
  for the vector $\vec{s}$ and complete all vectors inside the cluster
  for $\vec{s}$ according to $\vec{s}_+^c$. Applying the same
  procedure for all centers in $S$, we obtain a solution for
  $(M',k,r)$, as required.

  In the latter case, \ie, the case of $\PAIRCLUSq{}$, let $M_{D'}''$
  be a completion of $M_{D'}'$ leading to the (inclusion-wise minimal)
  solution $\PPP$ for $(M_{D'}'',k,r)$, \ie, a partition of $M_{D'}'$
  into at most $k$ clusters, and
  consider a cluster $P \in \PPP$. Let $\vec{m}$ be any vector in $P$;
  which exists since $\PPP$ is inclusion-wise minimal. Then
  $\HDIST(\vec{m},\vec{m}')\leq r$ for every other vector $\vec{m}'\in
  P$. Moreover, since $D'$ satisfies (C2) also
  $\HDIST(\vec{m}_+,\vec{m}_+')\leq r$ for the vectors $\vec{m}_+$ and
  $\vec{m}_+'$ in $M'$ corresponding to $\vec{m}$ and $\vec{m}'$,
  respectively. Hence both
  vectors can be completed to be equal in all other coordinates,
  \ie, the coordinates in $[d]\setminus
  \HSET(\vec{m}_+,\vec{m}_+')$; let $\vec{m}_+^c$ be one such
  completion for $\vec{m}_+$. Since we choose $\vec{m}_+'$
  arbitrarily this is true also for all other vectors in the
  cluster. Hence, completing all vectors corresponding to vectors in $P$
  according to $\vec{m}_+^c$ and obtain a cluster for $(M',k,r)$
  containing the same vectors. Applying the same
  procedure for all sets in $\PPP$, we obtain a solution for
  $(M',k,r)$, as required.
\end{proof}
}

\lv{
\subsection{A Generic Reduction}
\label{subsec:reduction}
Here, we present a generic construction that is used in several hardness proofs throughout the paper.

Let $G$ be a graph, where $V(G)=\{v_1, \ldots, v_n\}$ and $m=|E(G)|$, and let $deg(v_i) \leq n-1$ denote the degree of $v_i$ in $G$. Fix an arbitrary ordering ${\cal O}=(e_1, \ldots, e_m)$ of the edges in $E(G)$.
For each vertex $v_i \in V(G)$, define a vector $\vec{a_i} \in \{0, 1\}^m$ to be the incidence/characteristic vector of $v_i$ w.r.t.~${\cal O}$; that is, $\vec{a_i}[j]=1$ if $v_i$ is incident to $e_j$ and $\vec{a_i}[j]=0$ otherwise. Afterwards, expand the set of coordinates of these vectors by adding to each of them $n(n-1)$ ``extra'' coordinates, $n-1$ coordinates for each $v_i$, $i \in [n]$; we refer to the $n-1$ (extra) coordinates of $v_i$ as the ``private'' coordinates of $v_i$. For each $v_i$, $i \in [n]$, we will choose a number $x_i \in \{0, \ldots, n-1\}$, where the choice of the number $x_i$ will be problem dependent, and we will set $x_i$ many coordinates among the private coordinates of $v_i$ to 1, and all other extra private coordinates of $v_i$ to 0. Let $M=\{\vec{a_i} \mid i \in [n]\}$ be the set of expanded vectors, where $\vec{a_i} \in \{0, 1\}^{m+n(n-1)}$ for $i \in [n]$. We have the following straightforward observation:

\begin{observation}
\label{obs:genericreduction}
For each $v_i$, where $i \in [n]$, the number of coordinates in $\vec{a_i}$ that are equal to 1 is exactly $deg(v_i) + x_i$, and
two distinct vertices $v_i, v_j$ satisfy $\HDIST(\vec{a_i},\vec{a_j}) =deg(v_i)+x_i + deg(v_j) +x_j$ if $v_i$ and $v_j$ are nonadjacent in $G$ and $\HDIST(\vec{a_i},\vec{a_j}) =deg(v_i)+x_i + deg(v_j) +x_j -2$ if $v_i$ and $v_j$ are adjacent.
\end{observation}

Throughout the paper, we denote by ${\cal R}$ the polynomial-time reduction that takes as input a graph $G$ and returns the set of vectors $M$ described above.
}

\section{Clustering with Incomplete Data}
\label{sec:fpt-incomplete}

%

We will show that \IADCLUSq{} are fixed-parameter
tractable parameterized by $k+r+\comb(M)$. Our algorithmic results are achieved via kernelization: we will apply the irrelevant vector and irrelevant coordinate techniques to obtain an equivalent instance of size
upper bounded
by a function of $k+r+\comb(M)$. 

Note that this implies that also the variants
\IADCLUS{} for complete data  are
fixed-parameter tractable parameterized by only $k+r$ (and also have a
polynomial kernel) and, as we will show in
a later section, both parameters are indeed
required. To explain how we obtained our results, we will start by considering
the general procedures for complete data first and then provide the
necessary changes for the case of incomplete data. Throughout the
section we will assume that $(M,k,r)$ is the given instance of \IADCLUSq{}.
Recall that, when using the parameter $\comb(M)$, we will use the sets $T_M$ and $R_M$
(as defined in the preliminaries), where $T_M\subset [d]$, $R_M\subset M$, and $|T_M|+|R_M|=\comb(M)$, and such that all $\blank$'s in $M\setminus
R_M$ occur only in coordinates in~$T_M$.

\myparagraph{Informal description of the algorithm for complete data.}
To perform kernelization, we start by identifying and removing irrelevant vectors;
those are vectors that can be removed from the instance and safely
added back to any valid clustering of the reduced instance to yield a
valid clustering of the original instance. One caveat is that,
for \INCLUS{}, the removed vectors may serve as cluster centers, and
hence, such vectors will have to be represented in the reduced
instance; we will discuss later (below) how this issue is dealt with.
To identify irrelevant vectors, we first show that, for each vector, we
can compute a ``representative set'' of vectors of its $(\leq
r)$-neighborhood whose size is upper bounded by a function of the
parameter.
The identification of representative sets is achieved via a
non-trivial application of the Sunflower Lemma (and several other
techniques) in
Lemmas~\ref{lem:CLUS-SF-BASIC},~\ref{lem:CLUS-SF-NEIGH}
as well as
Lemma~\ref{lem:UN-CLUS-SF-COM-NEIGH} for \ANYCLUS{},
Lemma~\ref{lem:UN-CLUS-SF-COM-NEIGH} and~\ref{lem:UN-IN-POT-SOL-SET}
for \INCLUS{},
and Lemma~\ref{lem:DIAM-CLUS-SF-COM-NEIGH} for \ANYCLUS{}.
The union of these representative
sets yields a reduced instance whose number of vectors is upper
bounded by a function of the parameter. For the final step of our
algorithm we use our toolkit to reduce the
number of dimensions for every vector in the reduced instance. This is already sufficient to solve \ANYCLUS{}.

As for \INCLUS{}, we need to ensure that the centers of the clusters
in any valid solution are represented in the reduced
instance (whose size is now bounded by a function of the parameter).
To do so, we partition the
set of vectors removed from the reduced instance into
equivalence classes based on their ``trace'' on the set of
important coordinates; the number of equivalence classes is upper
bounded by a function of the parameter. Since each potential center
must be within distance $r$ from some
vector in the reduced instance, for each (irrelevant) vector $\vec{x}$
that differs in at most~$r$ important coordinates from some vector in
the reduced instance, we add a vector from the equivalence class of
$\vec{x}$ (that represents $\vec{x}$) whose distance to the vectors in
the reduced instance w.r.t.~nonimportant coordinates (which all
vectors in the reduced instance agree on) is
minimum. Lemma~\ref{lem:UN-IN-POT-SOL-SET} provides a bound on the number of these added vectors.


\myparagraph{Finding Redundancy when Data is Missing.}
In the case of incomplete data, we will in principle employ the same general strategy
that we used for clustering problems with complete data. Namely, we will again identify
irrelevant vectors and coordinates whose removal results in an
instance whose size can be bounded by our parameter. However, due to
the presence of incomplete data, we
need to make significant adaptations at every step of the algorithm.

Consider the first step of the algorithm, which allowed
us to identify and remove irrelevant vectors. For this
step, we can focus only on the vectors in $M\setminus R_M$,
since $|R_M|$ is already bounded by $\comb(M)$; crucially, this allows us to
assume that vectors only
have $\blank$\hy entries at positions in~$T_M$.

Now consider Lemma~\ref{lem:CLUS-SF-NEIGH}, which allowed
us to remove any vector, say $\vec{f}$, in a sufficiently large sunflower occurring in
the $t$-Hamming neighborhood of some vector $\vec{v}$. Informally, this was because in every solution of the reduced
instance, a large part of the sunflower must end up together in one of
the clusters; this in turn meant that for every vector in the cluster
there is a vector in the sunflower that is at least as far as $\vec{f}$.
This is what allowed us to
argue that $\vec{f}$ can always be safely added back into that
cluster. But this can no longer be guaranteed once $\blank$\hy entries are allowed, since whether $\vec{f}$ can be added back into the cluster or not depends on how the other vectors in the sunflower have been completed.
%

Note that
the problem above would disappear if we could ensure that a sufficiently
large number of vectors from
the initial sunflower that end up together in the same cluster
have the $\blank$\hy entries at the \emph{exact same positions}. Since we
observed earlier that we can assume that all vectors have their
$\blank$\hy entries only in $T_M$, and consequently there are at most
$2^{|T_M|}$ different allocations of the $\blank$\hy entries to these
vectors, we can now enforce this by enlarging the initial
sunflower by a factor of $2^{|T_M|}$. This approach allows us to obtain the
following lemma, which uses Lemma~\ref{lem:CLUS-SF-NEIGH} in a way that allows us to reduce the number of vectors for \INCLUSq{} and \ANYCLUSq.




\lv{\begin{lemma}}
  \sv{\begin{lemma}}
\label{lem:UN-CLUS-SF-COM-NEIGH}
  Let $k, r \in \Nat$ and $M \subseteq \{0,1,\blank\}^d$.
  Then there is a subset $M'$ of $M$ with $R_M \subseteq M'$ satisfying:
  \begin{itemize}
  \item(P1) For every $\vec{v} \in M\setminus R_M$ it holds that
    $|\HN{\vec{v}}{r}\cap M'\setminus R_M|\leq 2^{|T_M|}(\sum_{t=1}^{r}t!(k(r+t)+2)^{t})$; and
  \item(P2) for every set $S \subseteq \{0,1\}^d$ with $|S|\leq k$
    and satisfying $\HDIST(S,\vec{m})\leq r$ for every $\vec{m} \in M$
    it holds that $\max_{\vec{y} \in M}\HDISTq(S,\vec{y})=\max_{\vec{y} \in M'}\HDISTq(S,\vec{y})$.
  \end{itemize}
Moreover, $M'$ can be computed in time polynomial in $M$.
\end{lemma}
\lv{\begin{proof}

	We obtain $M'$
	using the following algorithm.
	Initially, we set $M'$ to $M$.
	Then for every $\vec{v} \in M\setminus R_M$, every $Q\subseteq T_M$, and every $t$ with $1 \leq t \leq r$ 
	we do the following.

  We denote by $M_Q$ the subset of $M\setminus R_M$ with $\{i\mid \vec{a}[i]=\blank\} = Q$ for every vector $\vec{a}\in M_Q$. Moreover, let $\sigma_Q: \{0,1,\blank\}^d\rightarrow \{0,1\}^{d-|Q|}$ be the mapping that given a vector $\vec{a}\in \{0,1,\blank\}^d$ outputs the vector $\sigma_Q(\vec{a})$ which skips all the coordinates of $\vec{a}$ in $Q$, or more formally, for all $i\in [d]$ and $j= |Q\cap [i]|$, we let $\sigma_Q(\vec{a})[i-j]=\vec{a}[i]$.
   Note that for a $\vec{a}\in M_Q$, the mapping $\sigma_Q$ skips exactly all the coordinates with $\blank$. Therefore, for every pair of vectors $\vec{a}\in M_Q$ and $\vec{s}\in \{0,1\}^d$ it holds $\HDISTq(\vec{a},\vec{s}) = \HDIST(\sigma_Q(\vec{a}),\sigma_Q(\vec{s}))$.

   We denote by $M'_Q$ the set $M'\cap M_Q$.
	 Now we apply Lemma~\ref{lem:CLUS-SF-NEIGH} to $\sigma_Q(\vec{v})$ and $\sigma_Q(M'_Q)$ exhaustively, \ie,
	as long as $|N|=|N_{=t}(M')|\geq t!(k(r+t+2))^t$, we use the lemma
	to find the vector $\vec{f}\in \sigma_Q(M'_Q)$, we remove from $M'$ the vector $\vec{g}\in M'_Q$ such that $\sigma_Q(\vec{g})=\vec{f}$, and apply the
	lemma again. Let $M'$ be the subset of $M$ obtained in this manner.
	Then (P1) clearly holds and (P2) follows from (P1)
	in Lemma~\ref{lem:CLUS-SF-NEIGH} and the observation that, for
        every pair of vectors $\vec{a}\in M_Q$ and $\vec{s}\in
        \{0,1\}^d$, it holds that $\HDISTq(\vec{a},\vec{s}) =
        \HDIST(\sigma_Q(\vec{a}),\sigma_Q(\vec{s}))$. Finally, $R_M
        \subseteq M'$ since we did not
        removed any vector in $R_M$ from $M$.
\end{proof}
}

Using the above Lemma~\ref{lem:UN-CLUS-SF-COM-NEIGH} together with our toolbox
(for reducing the number of relevant coordinates), we are now ready to show our first
fixed-parameter algorithm for \ANYCLUSq{}.
\lv{\begin{theorem}}
\sv{\begin{theorem}}
\label{thm:Q-any}
  \ANYCLUSq{} is \FPT\ parameterized by $k+r+\comb(M)$.
\end{theorem}
\lv{
  \begin{proof}
    Let $(M,k,r)$ be the given instance of \ANYCLUS{} and let $M'$ be
    the set obtained using Lemma~\ref{lem:UN-CLUS-SF-COM-NEIGH} for $M$,
    $k$, and $2r$. Because
    $M'$ satisfies (P2), it holds that $(M,k,r)$ and $(M',k,r)$ are
    equivalent instances of \ANYCLUS{}. Consider a solution $S\subseteq
    \{0,1\}^d$ for $(M',k,r)$, with $|S|\leq k$, and a vector $\vec{s} \in
    S$. Since we can assume
    that $S$ is minimal, it holds that $\HN{\vec{s}}{r}\cap M'\neq
    \emptyset$ for every $\vec{s} \in S$. Let $\vec{y} \in
    \HN{\vec{s}}{r}\cap M'$ be arbitrarily chosen. Then $\HN{\vec{s}}{r}\cap M'
    \subseteq \HN{\vec{y}}{2r}\cap M'$. Moreover, since $M'$ satisfies
    (P1), it follows that $\HN{\vec{y}}{2r}\cap M'$ and thus also
    $\HN{\vec{s}}{r}\cap M'$ has size at most
    $\sum_{t=1}^{2r}t!(k(2r+t+2))^{t}+1$. Consequently, if
    $|M'|>k((\sum_{t=1}^{2r}t!(k(2r+t+2))^{t})+1)$, we can safely
    return that $(M,k,r)$ is a \NO\hy instance of \ANYCLUSq{}.
    Thus,
    $|M'|\leq k((\sum_{t=1}^{2r}t!(k(2r+t+2))^{t})+1)$ and
    it remains to reduce the number of coordinates for each vector in
    $M'$. Let $D'$ be the set of coordinates obtained from
    Lemma~\ref{lem:impot-coor-equiv} for $M'$. Then $(M',k,r)$ and
    $(M_{D'}',k,r)$ are equivalent instances of $\ANYCLUSq{}$ and
    moreover we obtain from Lemma~\ref{lem:IN-Q-DIAM-BOUND} that
    $\MDIAMq(M')\leq 4rk-r+|T_M|$.
    Therefore, we obtain:

    \[\begin{array}{ccc}
      |D'| & \leq & (k\MDIAMq(M')+|R_M|(|M'|-1))(r'+1)\\
           & \leq & (k(4rk-r+|T_M|)+|R_M|(|M'|-1))(2r+1)
      \end{array}\]

    showing that the size of $D'$ is bounded by our parameter
    $k+r+\comb(M)$. Hence $(M_{D'}',k,r)$ is a kernel for
    $(M,k,r)$ and $\ANYCLUSq{}$ is fixed-parameter
    tractable parameterized by $k+r+\comb(M)$.
\end{proof}
}
Towards showing our kernelization result for \INCLUSq{}, we
need to add back some vectors that can be potential centers for the
clusters containing vectors of $M'$. The main idea for the case of
complete data is the observation that every vector in $M$ that can act
as a potential center for the instance on $M'$ must be within the
$r$-neighborhood of some vector in $M'$ and moreover among all
(potentially many vectors within the $r$-neighborhood of a vector in
$M'$), we can chose any vector, which is closest w.r.t. the unimportant
coordinates, \ie, the coordinates in $[d]\setminus \DCOOR(M')$. This
way the number of potential vectors that can act as a center for a
vector in $M'$ can be bounded by the parameter. For the case of
incomplete data we need to consider an additional complication,
namely, that the $\blank$
entries of the vectors in $M'$ (which can be changed without
increasing the Hamming distance to the vector), can increase the size of the
$r$-Hamming neighborhood of every such vector now significantly.
For instance, the potential $r$-Hamming neighborhood of a vector
in $M'\setminus R_M$ increases by a factor of $2^{|T_M|}$ and
the potential $r$-Hamming neighborhood of a vector $\vec{x}$ in $R_M$ can only be
bounded by $2^{\DIAM(M')(|M'|-1)}3^{|T_M|}$, since every
important coordinate of $\vec{x}$ could be a~$\blank$.

\lv{\begin{lemma}}
  \sv{\begin{lemma}}
\label{lem:UN-IN-POT-SOL-SET}
  Let $(M,k,r)$ be an instance of \INCLUSq{} and $M'\subseteq M$ with $R_M\subseteq M'$. Then
  there is a set $M''$ with $M' \subseteq M'' \subseteq M$ of size at
  most $|M'|+3^{|T_M|}r^{|R_M|}+k3^{2|T_M|}r^{|R_M|}2^{\MDIAMq(M')(|M'|-1)}$ such that
  there is a set $S\subseteq M$ with $|S|\leq k$ satisfying
  $\max_{\vec{y}\in M'}\HDISTq(S,\vec{y})\leq r$ if and only if there
  is a set $S\subseteq M''$ with $|S|\leq k$ satisfying
  $\max_{\vec{y}\in M'}\HDISTq(S,\vec{y})\leq r$. Moreover, $M''$ can
  be computed in polynomial time.
\end{lemma}
\lv{\begin{proof}
  First note that if $|R_M|=|M|$, then $M''=M$ trivially satisfies the
  conditions of the lemma.
  Let $S\subseteq M$ with $|S|\leq k$ satisfying $\max_{\vec{y}\in
    M'}\HDISTq(S,\vec{y})\leq r$ and let $\vec{s} \in S$. Since we can
  assume that $S$ is minimal (w.r.t.~satisfying $\max_{\vec{y}\in
    M'}\HDISTq(S,\vec{y})\leq r$), we obtain that $\vec{s}$
  must be within the $r$-Hamming neighborhood of some vector $\vec{y}
  \in M'$, \ie, for every $\vec{s} \in S$ there is a vector $\vec{y}
  \in M'$ such that $\HDISTq(\vec{s},\vec{y})\leq r$.

  We start by defining what it means for two vectors $\vec{m}$ and
  $\vec{m}'$ to be equivalent w.r.t. the vectors in $R_M$ in the sense
  that if $\vec{m}$ can be used as a center containing a vector in
  $R_M$ then so can $\vec{m}'$ and vice versa. Namely,
  we say that two vectors $\vec{m}$ and $\vec{m}'$ are
  \emph{equivalent w.r.t. $R_M$}, denoted $\equiv^{R_M}$,
  if and only if they agree on the coordinates in $T_M$, \ie,
  $\vec{m}[i]=\vec{m}[i]$ for every $i\in T_M$, and for every
  $\vec{x}$ in $R_M$ it holds that:
  \begin{itemize}
  \item $\HDIST(\vec{m},\vec{x})=\HDIST(\vec{m}',\vec{x})$ if
    $\HDIST(\vec{m},\vec{x})\leq r$ and
  \item $\HDIST(\vec{m}',\vec{x})>r$ otherwise.
  \end{itemize}
  Clearly, this guarantees that if a vector $\vec{x}$ is part of a
  cluster with center $\vec{m}$, then $\vec{x}$ will still be
  contained in the cluster if $\vec{m}$ is replaced by $\vec{m}'$ and
  completed in the same manner as $\vec{m}$. Note also that this
  defines an equivalence relation for the vectors in $M\setminus R_M$
  and that the number of equivalence classes is at most
  $3^{|T_M|}r^{|R_M|}$.

  We now have to define what it means for two vectors $\vec{m}$ and
  $\vec{m}'$ to be equivalent w.r.t. the remaining vectors in
  $M$, \ie, the vectors in $M\setminus R_M$. Towards this aim
  let $C$ be a component of $\compG(\inst)[M'\setminus R]$. Then it
  follows from Lemma~\ref{lem:diam-comp} that $|\DCOOR(C)\cup T_M|\leq
  \DIAMq(C)(|C|-1)+|T_M|$ and hence all vectors in $C$ agree on all
  coordinates outside of $|\DCOOR(C)\cup T_M|$.
  We say that two
  vectors $\vec{m}$ and $\vec{m}'$ in $M\setminus R_M$ that are in the
  $r$-Hamming neighborhood of some vector in $C$ are \emph{equivalent w.r.t. $C$}, denoted
  $\vec{m}\equiv_{C}\vec{m}'$, if and only if they agree on all
  coordinates in $\DCOOR(C)\cup T_M$. Then $\equiv_C$ is also an
  equivalence relation for all vectors that are in the $r$-Hamming
  neighborhood of some vector in $C$ and moreover the number of
  equivalence classes is at most $2^{|\DCOOR(C)\setminus
    T_M|}3^{|T_M|}$. Note, however, that even if
  $\vec{m}\equiv_C\vec{m}'$ and some vector $\vec{c} \in C$ is
  contained in a cluster with center $\vec{m}$, we cannot simple
  replace $\vec{m}$ by $\vec{m}'$ because $\vec{m}$ and $\vec{m}'$
  might have a different Hamming distance to $\vec{c}$ when
  considering the coordinates outside of $\DCOOR(C)\cup
  T_M$. Nevertheless, it still suffices to keep only one vector from every
  equivalence class, namely, a vector that is closest to any (all)
  vectors in $C$ w.r.t. the coordinates in $[d] \setminus
  (\DCOOR(C)\cup T_M)$; recall that all vectors in $C$ agree on all
  coordinates in $[d] \setminus (\DCOOR(C)\cup T_M)$. Note that
  $\equiv$ has at most $3^{|T_M|}r^{|R_M|}2^{|\DCOOR(C)\setminus
    T_M|}3^{|T_M|}=3^{2|T_M|}r^{|R_M}2^{|\DCOOR(C)\setminus T_M|}$
  equivalence classes.

  We can now combine these two equivalence relations into one.
  Namely, we say that two vectors $\vec{m}$ and
  $\vec{m}'$ in $M\setminus R_M$ are \emph{equivalent w.r.t. $C$ and $R_M$}, denoted by
  $\vec{m}\equiv_C^{R_M}\vec{m}'$, if and only if
  $\vec{m}\equiv^{R_M}\vec{m'}$ and $\vec{m}\equiv_C\vec{m'}$.

  Let $M_0$ be the set of all vectors (in $M\setminus M'$) defined as
  follows. First for every component $C$ of $\compG(\inst)[M']$ and every
  equivalence class $P$ of $\equiv_C^{R_M}$, the set $M_0$ contains
  a vector that is closest to any/every vector in $C$
  w.r.t. the coordinates in $[d]\setminus (\DCOOR(C)\cup T_M)$ among
  all vectors in $P$. Finally, we also
  add to $M_0$ an arbitrary vector for every equivalence class of
  $\equiv^{R_M}$.
  We claim that setting $M''$ to $M'\cup M_0$
  satisfies the claim of the lemma. We start by showing that
  there is a set $S\subseteq M$ with $|S|\leq k$ satisfying
  $\max_{\vec{m}\in M'}\HDISTq(S,\vec{m})\leq r$ if and only if there
  is a set $S\subseteq M''$ with $|S|\leq k$ satisfying
  $\max_{\vec{m}\in M'}\HDISTq(S,\vec{m})\leq r$. The reverse
  direction is trivial since $M''$ is a subset of $M$. Towards showing
  the forward direction let $S\subseteq M$ with $|S|\leq k$ satisfying
  $\max_{\vec{m}\in M'}\HDISTq(S,\vec{m})\leq r$ and let $\vec{s} \in
  S$. Then $\vec{s} \in M\setminus R_M$ since $R_M\subseteq M''$. Let $P_{\vec{s}}$
  be the set of all vectors in $M'$ that are in the cluster with
  center $\vec{s}$. If $P_{\vec{s}} \subseteq R_M$
  and we replace $\vec{s}$ with a vector $\vec{s}' \in M''$
  such that $\vec{s}\equiv^{R_M}\vec{s}'$. Since $\vec{s}$ and
  $\vec{s}'$ agree on all coordinates in $T_M$,
  we can complete $\vec{s}'$ in the same way as
  $\vec{s}$. Moreover, because $\HDIST(\vec{s},\vec{m})\leq r$ for
  every $\vec{m} \in P_{\vec{s}}$ and $\vec{m} \in R_M$, we obtain
  that $\HDIST(\vec{s}',\vec{m})=\HDIST(\vec{s},\vec{m})\leq r$, as required.
  Otherwise, let $\vec{m} \in P_{\vec{s}}$ be a vector with $\vec{m}
  \in M'\setminus R_M$ and let $C$ be the component of
  $\compG(\inst)[M']$ containing $\vec{m}$. Then $P_{\vec{s}}\setminus
  R_M\subseteq C$. We claim that we can replace $\vec{s}$ with the
  vector $\vec{s}'$ in $M''$ such that
  $\vec{s}'\equiv_C^{R_M}\vec{s}$. Since $\vec{s}'$ agrees with
  $\vec{s}$ on all coordinates in $T_M$, we can complete $\vec{s}'$ in
  the same manner as $\vec{s}$.
  Let $\vec{x}$ be a vector in $P_{\vec{s}}$. If $\vec{x} \notin R_M$,
  then $\HDIST(\vec{s},\vec{x})\leq \HDIST(\vec{s}',\vec{x})$, because
  $\vec{s}$ and $\vec{s}'$ agree on all coordinates in $\DCOOR(C) \cup
  T_M$ and $\vec{s}'$ is closest to all vectors in $C$ (and in
  particular to $\vec{x}$) w.r.t.~to all other coordinates. Moreover, if
  $\vec{x} \in R_M$, then $\HDIST(\vec{s},\vec{x})\leq r$ and hence
  $\HDIST(\vec{s}',\vec{x})=\HDIST(\vec{s},\vec{x})\leq r$, as required.

  We are now ready to bound the size of $M_0$ is terms of our
  parameter $k+r+\comb(M)$ and the size of $M'$. Apart from the $3^{|T_M|}r^{|R_M|}$
  vectors (one for every equivalence class of $\equiv^{R_M}$), $M_0$
  contains $3^{2|T_M|}r^{|R_M|}2^{|\DCOOR(C)\setminus T_M|}$ vectors
  for every component $C$ of $\compG(\inst)[M']$ (one for every
  equivalence class of $\equiv_C^{R_M}$). Since $\compG(\inst)[M']$
  can have at most $k$ components, since otherwise $\inst$ is a
  \NO\hy instance,  and using the fact that $|\DCOOR(C)\setminus T_M|\leq
  \MDIAMq(M')(|M'|-1)$ for every component $C$ of $\compG(\inst)[M']$ (Lemma~\ref{lem:diam-comp}), we obtain that:

  \[
    |M_0|\leq
    3^{|T_M|}r^{|R_M|}+k3^{2|T_M|}r^{|R_M|}2^{\MDIAMq(M')(|M'|-1)}.
  \]

\noindent   It is straightforward to verify that $M_0$ can be computed in
  polynomial time.
  \end{proof}
}

With Lemma~\ref{lem:UN-IN-POT-SOL-SET} in hand, we can establish the fixed-parameter tractability of \INCLUSq{}.
\lv{\begin{theorem}}
\sv{\begin{theorem}}
\label{the:any-q-kernel}
  \INCLUSq{} is \FPT\ parameterized by $k+r+\comb(M)$.
\end{theorem}
\lv{
\begin{proof}
  Let $(M,k,r)$ be the given instance of \INCLUS{} and let $M'$ be
  the set obtained using Lemma~\ref{lem:UN-CLUS-SF-COM-NEIGH}. Since
  $M'$ satisfies (P2), it holds that $(M,k,r)$ has a solution if and
  only if there is a set $S\subseteq M$ with $|S|\leq k$
  such that $\max_{\vec{y} \in M'}\HDISTq(S,\vec{y})\leq r$. Since
  $M'$ satisfies (P1), we can safely return that $(M',k,r)$ is a \NO\hy
  instance if
  $|M'|>k(2^{|T_M|}(\sum_{t=1}^{r}t!(k(r+t)+1)^{t})+|R_M|+1) = f(k,r,\comb(M))$. Hence,
  w.l.o.g., we can assume that $|M'|\leq
  f(k,r,\comb(M))$, and it only remains to
  bound the number of vectors in $M\setminus M'$ that could
  potentially be in a solution. Let $M''$ be the set obtained from
  Lemma~\ref{lem:UN-IN-POT-SOL-SET} for $M$ and $M'$. Then $(M,k,r)$ and
  $(M''\cup M',k,r)=(M'',k,r)$ are equivalent instances of \INCLUSq{}. Moreover,
  $|M''|\leq |M'|+3^{|T_M|}r^{|R_M|}+k3^{2|T_M|}r^{|R_M|}2^{\MDIAMq(M')(|M'|-1)}$,
  which together with
  Lemma~\ref{lem:IN-Q-DIAM-BOUND} implies that $|M''|$ is also bounded
  by a function of $k+r+\comb(M)$.
  Finally,
  it remains to reduce the number of coordinates for each vector in
  $M''$. Let $D'$ be the set of coordinates obtained from
  Lemma~\ref{lem:impot-coor-equiv} for $M''$. Then $(M'',k,r)$ and
  $(M_{D'}'',k,r)$ are equivalent instances of $\INCLUSq{}$ and
  moreover we obtain from Lemma~\ref{lem:IN-Q-DIAM-BOUND} that
  $\MDIAMq(M'')\leq 3rk-r+|T_M|$.
  Therefore, we obtain:

  \[\begin{array}{ccc}
      |D'| & \leq & (k\MDIAMq(M'')+|R_M|(|M''|-1))(r'+1)\\
           & \leq & (k(3rk-r+|T_M|)+|R_M|(|M''|-1))(r+1)
      \end{array}\]

  showing that the size of $D'$ is bounded by our parameter
  $k+r+\comb(M)$. Hence, $(M_{D'}'',k,r)$ is a kernel for
  $(M,k,r)$ and $\INCLUSq{}$ is fixed-parameter
  tractable parameterized by $k+r+\comb(M)$.
\end{proof}
}

We now proceed to the last of the three problems considered in this section, \PAIRCLUSq{}. Apart
from the issue that we already had for \INCLUSq{} and \ANYCLUSq{}
that we require a sunflower of vectors with all $\blank$s in the
same position, we now have the additional complication that we can
no longer assume that the $\blank$\hy entries of vectors that end up
in the same cluster are completed in the same way;
note that this is not an issue for \INCLUSq{} and \ANYCLUSq{} since
there one can always assume that all elements in a cluster are
completed the same way as the center vector. We show that
this problem can be handled by increasing the size of the sunflower by an additional
factor of $2^{|T_M|}$. Because of the same issue, we also need to
take into account the potential distance between different vectors
in the same cluster arising from the possibility of different
completions of the coordinates in $T_M$. This leads to the following
version of Lemma~\ref{lem:UN-CLUS-SF-COM-NEIGH} for \PAIRCLUSq{}.
\lv{
\begin{lemma}\label{lem:UN-CLUS-SF-NEIGH-DIAM}
  Let $k, r, t \in \Nat$, $M \subseteq \{0,1,\blank\}^d$,
  $\vec{v} \in M$, 
  and let $N:=\HN{\vec{v}}{=t}\cap M\setminus R_M$. If $|N|\geq
  2^{|T_M|}t!(2^{|T_M|}k(r+t+|T_M|+2))^{t}+1$, then there is a vector $\vec{f} \in N$
  satisfying the following property:
  \begin{quote}
   $M$ has a completion $M'\subseteq \{0,1\}^d$ with a partition
    into at most $k$ clusters, each of diameter at most $r$, if and only if
    $M\setminus \{\vec{f}\}$ does.
  \end{quote}
  Moreover, $\vec{f}$ can be determined in time polynomial in $M$.
\end{lemma}
\begin{proof}
	\newcommand{\HCL}{\Gamma}
	Let us first define an equivalence relation $\sim$ over $N\setminus R_M$ depending on which subset of $T_M$ contains $\blank$. That is, for two vectors $\vec{a},\vec{b}\in N$, we say $\vec{a}\sim\vec{b}$ if and only if $\{i\mid \vec{a}[i]=\blank\} = \{i\mid \vec{b}[i]=\blank\}$.
	
	Now let $N^\blank$ be a maximum size equivalence class of $\sim$, and let $\FFF:=\SB \HSET(\vec{v},\vec{x}) \SM \vec{x} \in N^\blank \SE$. Then
	$|\FFF|=|N^\blank|\geq t!(2^{|T_M|}k(r+t+|T_M|+2))^{t}$, $|F|=t$ for every $F \in
	\FFF$, and all the vectors of $N^\blank$ have $\blank$\hy entries at the same indices. By Lemma~\ref{lem:SF}, $\FFF$ contains a
	sunflower, say $\FFF'$, of size at least $2^{|T_M|}k(r+t+|T_M|+2)$ (and core $C$).
	
	We denote by $N^\blank(F)$
	the vector in $N^\blank$ giving rise to the element $F \in \FFF$, \ie,
	$F=\HSET(\vec{v},N^\blank(F))$. Moreover, for a subset $\FFF'$ of
	$\FFF$ we denote by $N^\blank(\FFF')$ the set $\SB N^\blank(F) \SM F \in \FFF'\SE$.
	Let $F \in \FFF'$ be arbitrarily chosen. We claim that setting $\vec{f}$ to the vector
	$N^\blank(F)$ satisfies the claim of the lemma.
	Since $|\FFF'\setminus \{F\}|\geq 2^{|T_M|}k(r+t+|T_M|+2)$, we observe the following:
	\begin{quote}
		for every partition $\PPP$ of $M\setminus \{\vec{f}\}$ into at most $k$ sets
		there is a set $P \in \PPP$ with $|P \cap N^\blank(\FFF'\setminus \{F\})|\geq 2^{|T_M|}(r+t+|T_M|+1)+1$.
	\end{quote}
	
	We are now ready to prove the lemma. First note that the forward direction of the lemma holds trivially.
Towards showing the other direction, let $M'_{\vec{f}}\subseteq \{0,1\}^d$ be a completion of $M\setminus \vec{f}$ with the bijection $\alpha:M\setminus\vec{f}\rightarrow M'_{\vec{f}}$ as the completion witness, and let $\PPP$ be any
	partition of $M'_{\vec{f}}$ into at most $k$ sets, each of
	diameter at most $r$. By the above observation, there is a set $P \in \PPP$
	with $|\alpha^{-1}(P)\cap N^\blank(\FFF'\setminus \{F\})|\geq
        2^{|T_M|}(r+t+|T_M|+2)$. Note that all the vectors in
        $N^\blank(\FFF'\setminus \{F\})$ have $\blank$ exactly at the
        indices in $T_M$. Hence there
        is a set $P'\subseteq P$ of size $r+t+|T_M|+2$ which was
        completed the same way, \eg, $\vec{a}[i]=\vec{b}[i]$ for all
        $\vec{a},\vec{b}\in P'$ and all $i$ such that
        $\alpha^{-1}(\vec{a})[i]=\alpha^{-1}(\vec{b})[i]=\blank$. Let
        $\alpha(\vec{f})$ be the completion of $\vec{f}$ in the same
        way as all the other vectors in $P'$. Note that, since for all
        $\vec{a}\in \alpha^{-1}(P')$ it holds that
        $\HDIST(\vec{a},\vec{v}) = t$ and all the vectors in
        $\alpha^{-1}(P')$ are completed the same way, it follows that
        there is an integer $t'$ with $t-|T_M|\le t'\le t+|T_M|$ such
        that, for every vector $\vec{a}\in \alpha^{-1}(P')$, we have
        $\HDIST(\vec{a},\alpha(\vec{v}))=t$. Moreover, $\FFF_{P'}:=\SB
        \HSET(\alpha(\vec{v}),\vec{x}) \SM \vec{x} \in P' \SE$ is a
        sunflower with $|\FFF_{P'}|\ge r+t'+2$.

	Let $N'$ be the
	set $(P'\cap \alpha(N^\blank(\FFF'\setminus \{F\})) \cup \{\alpha(\vec{f})\}$.
	Then $t'$, $r$,
	$\alpha(\vec{v})$, $N'$, $N'$
	satisfy the conditions of Lemma~\ref{lem:CLUS-SF-BASIC}. By
	observing that $\HDIST(N',\vec{p})\leq r$ for every $\vec{p}\in P$,
	we obtain that $\HDIST(\alpha(\vec{f}),\vec{p}) \leq \max_{x \in N'\setminus
		\{\alpha(\vec{f})\}}\HDIST(\vec{x},\vec{p})$ for every $\vec{p} \in
	P$. Hence $P\cup \{\alpha(\vec{f})\}$ has diameter at most $r$, which
	implies that the partition obtained from $\PPP$ after adding $\alpha(\vec{f})$
	to $P$ is a partition of $M'_{\vec{f}}\cup \alpha(\vec{f})$, which is a completion of $M$, into at most $k$ clusters, each of diameter at most $r$.
\end{proof}
}

\lv{\begin{lemma}}
\sv{\begin{lemma}}
\label{lem:DIAM-CLUS-SF-COM-NEIGH}
	Let $k, r \in \Nat$, and $M \subseteq \{0,1,\blank\}^d$.  
	Then there is a subset $M'$ of $M$ with $R_M \subseteq M'$ satisfying:
	\lv{
	\begin{enumerate}[leftmargin=*,widest=(P1)]
        \item[(P1)] For every $\vec{v} \in M\setminus R_M$ it holds that
          $|\HN{\vec{v}}{r}\cap (M'\setminus R_M)|\leq
          2^{|T_M|}(\sum_{t=1}^{r+|T_M|}t!(2^{|T_M|}k(r+t+|T_M|)+2)^{t})+1$; and
        \item[(P2)] $M$ has a completion with a
          partition into at most $k$ clusters of diameter at most $r$
          if and only if $M'$ does.
	\end{enumerate}
	}
	\sv{\textbf{\textup{(P1)}} For every $\vec{v} \in M\setminus R_M$ it holds that
          $|\HN{\vec{v}}{r}\cap (M'\setminus R_M)|\leq
          2^{|T_M|}(\sum_{t=1}^{r+|T_M|}t!(2^{|T_M|}k(r+t+|T_M|)+2)^{t})+1$; and
        \textbf{\textup{(P2)}} $M$ has a completion with a
          partition into at most $k$ clusters of diameter at most $r$
          if and only if $M'$ does.	}Moreover, $M'$ can be computed in time polynomial in $M$.
\end{lemma}
\lv{
\begin{proof}
	Initially, we set $M'$ to $M$.
	Then
	for every $\vec{v} \in M\setminus R_M$ and every $t$ with $1 \leq t \leq r+|T_M|$, we
	apply Lemma~\ref{lem:UN-CLUS-SF-NEIGH-DIAM} to $\vec{v}$ and $M'$ exhaustively, \ie,
	as long as $|N|=|N_{=t}(M')|\geq 2^{|T_M|}t!(2^{|T_M|}k(r+t+|T_M|+2))^{t}+1$, we use the lemma
	to find the vector $\vec{f}$, remove it from $M'$ and apply the
	lemma again. Let $M'$ be the subset of $M$ obtained in this manner.
	Then $R_M \subseteq M'$ and (P1) clearly holds. Moreover,
        (P2) follows from Lemma~\ref{lem:UN-CLUS-SF-NEIGH-DIAM}.
\end{proof}
}

We can now prove that \PAIRCLUSq{} is \FPT\ w.r.t.\ the three parameters.
\lv{\begin{theorem}}
\sv{\begin{theorem}}
  \PAIRCLUSq{} is \FPT\ parameterized by $k+r+\comb(M)$.
\end{theorem}
\lv{
\begin{proof}
  Let $(M,k,r)$ be the given instance of \PAIRCLUSq{} and let $M'$ be
  the set obtained using Lemma~\ref{lem:DIAM-CLUS-SF-COM-NEIGH}. Because
  $M'$ satisfies (P2), it holds that $(M,k,r)$ and $(M',k,r)$ are
  equivalent instances of \PAIRCLUSq{}. Moreover, because $M'$
  satisfies (P1), we obtain that every cluster of diameter at most $r$
  can contain at most $(2^{|T_M|}(\sum_{t=1}^{r+|T_M|}t!(2^{|T_M|}k(r+t+|T_M|)+1)^{t})+|R_M|+1)$ vectors,
  which implies that if $|M'|>	k(2^{|T_M|}(\sum_{t=1}^{r+|T_M|}t!(2^{|T_M|}k(r+t+|T_M|)+1)^{t})+|R_M|+1)$, we
  can safely return that $(M,k,r)$ is a \NO\hy instance. Consequently,
  $|M'|\leq k((\sum_{t=1}^rt!(k(r+t+2))^{t})+1)$
  and
  it remains to reduce the number of coordinates for each vector in
  $M'$. Let $D'$ be the set of coordinates obtained from
  Lemma~\ref{lem:impot-coor-equiv} for $M'$. Then $(M',k,r)$ and
  $(M_{D'}',k,r)$ are equivalent instances of $\PAIRCLUSq{}$ and
  moreover we obtain from Lemma~\ref{lem:IN-Q-DIAM-BOUND} that
  $\MDIAMq(M')\leq 2rk-r+|T_M|$.
  Therefore, we obtain:

  \[\begin{array}{ccc}
      |D'| & \leq & (k\MDIAMq(M')+|R_M|(|M'|-1))(r'+1)\\
           & \leq & (k(2rk-r+|T_M|)+|R_M|(|M'|-1))(r+1)
      \end{array}\]
  showing that the size of $D'$ is bounded by our parameter
  $k+r+\comb(M)$. Hence $(M_{D'}',k,r)$ is a kernel for
  $(M,k,r)$ and $\INCLUSq{}$ is fixed-parameter
  tractable parameterized by $k+r+\comb(M)$.
\end{proof}
}
\lv{
\section{Lower-Bound Results}
\label{sec:hardness}
We dedicate this section to showing that the parameterizations used in the presented \FPT{} algorithms presented in Section~\ref{sec:fpt-incomplete}, are necessary to achieve tractability. Obviously, lower-bound results for clustering problems for complete data carry over to their counterparts for incomplete data. Therefore, we will omit restating these results for the incomplete data case.

\subsection{Lower-Bound Results for Complete Data}
\label{subsec:lowerboundcomplete}
It is known that \ANYCLUS{} is \NP-complete for $r=2$~(see Section 3
of previous work by Jiao \etal~(\citeyear{JiaoXuLi04})). Our first
hardness results show that the other two clustering problems are also
\NP-complete for constant values of $r$. \sv{The proof is based on two
  reductions: one from the {\sc Dominating Set} problem on 3-regular
  graphs~\cite{kikuno1980np,gj}, and the other from the problem of
  partitioning a $K_4$-free 4-regular graph into
  triangles~\cite{triangles}.
}

\lv{\begin{theorem}}
\sv{\begin{theorem}}
\label{thm:NPh-r}
  \INCLUS{} is \NP-complete for $r=4$,
and \PAIRCLUS{} is \NP-complete for $r=6$.
\end{theorem}

%
\lv{\begin{proof}
For \INCLUS{}, we give a polynomial-time reduction from the {\sc Dominating Set} problem on 3-regular graphs (3-DS), which is \NP-complete~\cite{kikuno1980np,gj}, to the restriction of \INCLUS{} to instances where $r\leq 4$.
Given an instance $(G, k)$ of 3-DS, where $V(G)=\{v_1, \ldots, v_n\}$, set $x_i=0$ for $i \in [n]$, and apply the reduction ${\cal R}$ (described in Subsection~\ref{subsec:reduction}) to $G$ to obtain the set of vectors $M$.  By Observation~\ref{obs:genericreduction}, for any two vertices $v_i, v_j \in V(G)$, where $i \neq j$, we have $\HDIST(\vec{a_i},\vec{a_j}) =6$ if $v_iv_j \notin E(G)$ and $\HDIST(\vec{a_i},\vec{a_j}) =4$ if $v_i$ and $v_j$ are adjacent. The reduction from 3-DS to \INCLUS{} produces the instance $(M, k, 4)$ of \INCLUS{}.

It is easy to see that if $D$ is dominating set of $G$ of size $k$, then we can cluster $M$ into $k$ clusters, each containing a vector $\vec{a_i}$ corresponding to a vertex $v_i \in D$ and vectors corresponding to neighbors of $v_i$ in $G$; if a vertex $v_j$ in $G$ has multiple neighbors in $D$, then pick a neighbor $v_i \in D$ of $v_j$ arbitrarily, and place $\vec{a_j}$ in the cluster containing $\vec{a_i}$. Since each cluster $C$ contains a vector corresponding to a vertex $v_i$ in $D$, and all other vectors in $C$ correspond to neighbors of $v_i$ in $G$, the distance between any vector in $C$ and $\vec{a_i}$ is at most 4. This shows that $(M, k, 4)$ is a \yes-instance of \INCLUS{}. Conversely, if $(M, k, 4)$ is a \yes-instance of \INCLUS{}, let $C_1, \ldots, C_k$ be a partitioning of $M$ into $k$ clusters, each of radius at most 4, and let $\vec{a_{j_1}}, \ldots, \vec{a_{j_k}}$ be their centers, respectively. Consider the set of vertices $D=\{v_{j_1}, \ldots, v_{j_k}\}$. For a vertex $v_i \in V(G)$, its vector $\vec{a_i}$ belongs to a cluster $C_p$, $p\in [k]$, and hence its distance from $\vec{a_{j_p}}$ is at most 4. This implies that either $v_i=v_{j_p}$, or $v_i$ is adjacent to $v_{j_p}$ in $D$. It follows that $(G, k)$ is a \yes-instance of 3-DS.

 \medskip

The proof for \PAIRCLUS{} uses similar ideas, but the starting point of the reduction is different. Here, we reduce from the problem of determining whether a $K_4$-free 4-regular graph can be partitioned into triangles; we will simply refer to this problem as $\triangle$-{\sc Partition}.

We begin by arguing the \NP-hardness of $\triangle$-{\sc Partition}. First, Theorem~10 in~\cite{triangles} establishes the \NP-hardness of determining whether a graph of maximum degree $4$ can be partitioned into triangles. By Lemma~3 in~\cite{triangles}, this problem then admits a polynomial-time reduction to determining whether a $4$-regular graph can be partitioned into triangles. Finally, if a $4$-regular graph contains a $K_4$, then in any partitioning of the graph into triangles, the vertices of the $K_4$ belong to two triangles, $T_1, T_2$, such that $T_1$ consists of an edge $e_1$ of the $K_4$ and a vertex $v_1$ not in the $K_4$, and $T_2$ consists of an edge $e_2$ of the $k_4$ (such that $e_1$ and $e_2$ share no endpoints) and a vertex $v_2$ not in the $K_4$. Notice that in such case $T_1$ and $T_2$ are unique and can be determined and removed by pre-processing the instance. The above combined implies the \NP-hardness of our $\triangle$-{\sc Partition} problem.

Now given an instance $G$ of $\triangle$-{\sc Partition}, where $V(G)=\{v_1, \ldots, v_n\}$, we again set $x_i=0$ for $i \in [n]$, and apply the reduction ${\cal R}$ to $G$ to obtain the set of vectors $M$. The polynomial-time reduction from $\triangle$-{\sc Partition} to  \PAIRCLUS{}  produces the instance $(M, n/3,6)$ of \PAIRCLUS{}. By Observation~\ref{obs:genericreduction}, since $G$ is 4-regular, we have: for any two distinct vertices $v_i, v_j \in V(G)$, $\HDIST(\vec{a_i},\vec{a_j}) =8$ if $v_i$ and $v_j$ are nonadjacent and $\HDIST(\vec{a_i},\vec{a_j}) =6$ if $v_i$ and $v_j$ are adjacent.

If $G$ can be partitioned into $n/3$ triangles, then the three vectors in $M$ corresponding to each triangle form a cluster of diameter 6, and hence $(M, n/3,6)$ is a \yes-instance of \PAIRCLUS{}. Conversely, if
$(M, n/3,6)$ is a \yes-instance of \PAIRCLUS, then since $G$ is $K_4$-free, no cluster can contain more than three vectors. Since $M$ can be partitioned into $n/3$ clusters, it follows that each cluster contains exactly three vectors. The three vertices in $G$ corresponding to the three vectors in any of the $n/3$ clusters are pairwise adjacent, and hence form a triangle in $G$. It follows that $G$ can be partitioned into $n/3$ triangles, and $G$ is a \yes-instance of $\triangle$-{\sc Partition}.
\end{proof}
}

Having ruled out fixed-parameter tractability when parameterizing only by $r$, we turn to the case where the parameter is $k$ alone. First of all, for $k=1$ \ANYCLUS{} is equivalent to \textsc{Closest String}, a well-studied \NP-complete problem~\cite{GrammNiedermeierRossmanith03}.
\lv{Below we show that \PAIRCLUS{} is also \NP-complete even when restricted to a fixed value of $k$. }\sv{Using a two-step reduction from \textsc{3-Coloring} on 4-regular graphs~\cite{3coloring}, we show that \PAIRCLUS{} is also \NP-complete, even when restricted to a fixed value of $k$.}
\lv{\begin{theorem}}
\sv{\begin{theorem}}
\label{thm:paranprpairclusr}
\PAIRCLUS{} is \NP-complete for \hbox{$k=3$}.
\end{theorem}

\lv{
\begin{proof}
Consider the problem of deciding whether a graph on $n$ vertices can be partitioned into three cliques, referred to as {\sc 3-Clique Partitioning} henceforth. This problem is \NP-hard via a trivial reduction (that complements the edges of the graph) from the \NP-hard problem~\cite{3coloring} {\sc 3-Coloring}. We can now reduce {\sc 3-Clique Partitioning} to the restriction of  \PAIRCLUS{} to instances in which the number of desired clusters, $k$, is 3, using the generic construction given in Subsection~\ref{subsec:reduction}.

 Given an instance $G$ of {\sc 3-Clique Partitioning}, where $V(G)=\{v_1, \ldots, v_n\}$, we set $x_i=n-1-deg(v_i)$ for $i \in [n]$, and apply the polynomial-time reduction ${\cal R}$ to $G$ to produce the set of vectors $M$.
 The polynomial-time reduction from {\sc 3-Clique Partitioning} produces the instance $(M, 3, 2n-4)$ of \PAIRCLUS{}. By Observation~\ref{obs:genericreduction}, for any two distinct vertices $v_i, v_j \in V(G)$, $\HDIST(\vec{a_i},\vec{a_j}) =2n-2$ if $v_i$ and $v_j$ are nonadjacent and $\HDIST(\vec{a_i},\vec{a_j}) =2n-4$ if $v_i$ and $v_j$ are adjacent.

If $G$ can be partitioned into 3 cliques, then the vectors in $M$ corresponding to the vertices in each clique form a cluster of diameter $2n-4$, and hence $(M, 3, 2n-4)$ is a \yes-instance of \PAIRCLUS{}. Conversely, if
$(M, 3, 2n-4)$ is a \yes-instance of \PAIRCLUS{}, $M$ can be partitioned into $3$ clusters, each of diameter at most $2n-4$. The vertices in $G$ corresponding to the vectors in each of the 3 clusters are pairwise adjacent, and hence form a clique in $G$. It follows that $G$ can be partitioned into 3 cliques, and $G$ is a \yes-instance of {\sc 3-Clique Partitioning}.
\end{proof}
}

Finally, we note that, unlike the previous two problems, \INCLUS{}
admits a simple polynomial-time brute-force algorithm for every fixed
value of $k$ where the order of the polynomial depends on
$k$. However, one can still exclude fixed-parameter tractability
via a reduction from \textsc{Dominating Set}.

\sv{
\begin{theorem}]
\label{the:in-clus-k}
  \INCLUS{} is \Weft\emph{[2]}-complete parameterized by $k$ and can be solved in time $\bigoh(|M|^k|M|kd)$. Moreover, there is no algorithm solving \INCLUS{} in time $|M|^{o(k)}$ unless the Exponential Time Hypothesis fails.
\end{theorem}
%
}
\lv{
\begin{observation}
\label{obs:inclusterxp}
  \INCLUS{} can be solved in time $\bigoh(|M|^k|M|kd)$.
\end{observation}

\begin{proof}
The result follows using a brute-force algorithm that enumerates each subset of $k$ vectors in $M$ as the potential centers of the $k$ clusters sought. For each such subset $S \subseteq M$ of $k$ vectors, the algorithm iterates through the vectors in $M$, placing each vector $\vec{a} \in M$ into the cluster containing the vector in $S$ whose distance to $\vec{a}$ is minimum and is at most $r$; if no vector in $S$ has distance at most $r$ to $\vec{a}$, the enumeration is discarded, as it does not lead to a solution. If the algorithm manages to place each vector $\vec{a} \in M$ into a cluster containing a vector in $S$ whose distance to $\vec{a}$ is at most $r$, the algorithm accepts. Enumerating all subsets of $k$ vectors in $M$ takes time $\bigoh(|M|^k)$. Iterating through each vector in $M$, and finding its closest vector in the enumerated $k$-subset of $M$, takes time $\bigoh(|M|kd)$. The theorem follows.
\end{proof}

\begin{theorem}\label{the:in-clus-k}
  \INCLUS{} is \Weft\emph{[2]}-complete parameterized by $k$. Moreover, there is no algorithm solving \INCLUS{} in time $|M|^{o(k)}$ unless the Exponential Time Hypothesis fails.
\end{theorem}
\begin{proof}
We prove the statement by giving a reduction from {\sc Dominating Set} (DS), which is \W{2}-hard and cannot be solved in subexponential time unless the Exponential Time Hypothesis fails~\cite{DowneyFellows13}, to \INCLUS{} parameterized by $k$.
The reduction is very similar to that in the proof of Theorem~\ref{thm:NPh-r}, albeit that its starting point is {\sc Dominating Set} (on general graphs) rather than 3-DS. Given an instance $(G, k)$ of DS, where $V(G)=\{v_1, \ldots, v_n\}$, we set $x_i=n-1-deg(v_i)$ for $i \in [n]$, and apply the polynomial-time reduction ${\cal R}$ to $G$ to produce the set of vectors $M$.  The reduction from DS to \INCLUS{} produces the instance $\inst=(M, k, 2n-4)$ of \INCLUS{}.

By Observation~\ref{obs:genericreduction}, for any two distinct vertices $v_i, v_j \in V(G)$, $\HDIST(\vec{a_i},\vec{a_j}) =2n-2$ if $v_i$ and $v_j$ are nonadjacent and $\HDIST(\vec{a_i},\vec{a_j}) =2n-4$ if $v_i$ and $v_j$ are adjacent.
The proof that $(G, k)$ is a \yes-instance of DS iff $(M, k, 2n-4)$ is a \yes-instance of \INCLUS{} now follows by similar arguments to those in the proof of the same statement in Theorem~\ref{thm:NPh-r}.

Finally, membership in \W{2} can be shown via a reduction from \INCLUS{} to \textsc{DS} that constructs the compatibility graph $\compG(\inst)$ of the given instance $\inst$ in polynomial time, and uses the observation that there is a direct correspondence between a dominating set in $G$ of size $k$ and a solution for $\inst$.
\end{proof}
}

\subsection{Lower-Bound Results for Incomplete Data}
\label{subsec:lowerboundincomplete}
The earlier results in this section already show that out of the three
considered parameters, $k$ and $r$ must both be used if one wishes to
obtain fixed-parameter algorithms for the clustering problems under consideration. In the
case of clustering of incomplete data, the only two questions that
remain are whether one also needs to use the covering number\lv{
$\comb(M)$}, and whether it is possible to extend the polynomial-time
algorithm for \INCLUS\ to \INCLUSq.
\sv{We resolve these last remaining questions
using reductions from \probfont{$3$-Coloring} and \probfont{Closest String}.}
\lv{We resolve these questions below.}
\sv{\begin{theorem}
  \label{thm:quest-NPh-kar}
  For $X\in \{\probfont{In,Any,Diam}\}$,  \CLUSq{$X$} is  \NP-complete even if $k=3$
  and $r=0$. Furthermore, \INCLUSq{} is \NP-complete even if $k=1$
  and there is only one row containing $\blank$\hy entries.
\end{theorem}
}
\begin{theorem}
  \label{thm:quest-NPh-kar}
  \INCLUSq{}, \ANYCLUSq{}, \PAIRCLUSq{} are \NP-complete even if $k=3$
  and $r=0$.
\end{theorem}
\begin{proof}
  We give a polynomial-time reduction from \probfont{$3$-Coloring} as
  follows. Let $G$ be the given instance of \probfont{$3$-Coloring}
  with edges $e_1,\dotsc,e_{|E(G)|}$
  and let $M$ be the set of vectors containing a vector $\vec{v} \in
  \{0,1,\blank\}^{|E(G)|}$ for every $v \in V(G)$ such that
  $\vec{v}[i]=\blank$ if $e_i$ is not incident with $v$, $\vec{v}[i]=0$ if $e_i=\{u,v\}$
  is incident with $v$ and $v<u$, and $\vec{v}[i]=1$ otherwise; here we
  assume an arbitrary but fixed ordering $<$ of the vertices of $G$.
  It is now straightforward to verify that $G$ has a $3$-coloring if
  and only if the vectors in $M$ can be partitioned into three sets
  such that $\HDISTq(x,y)=0$ for every $x,y \in M$
  contained in the same set, which in turn is true if and only if $(M,3,0)$
  is a \YES\hy instance of \INCLUSq{}, \ANYCLUSq{}, or \PAIRCLUSq{}.
\end{proof}

 \begin{theorem}
   \label{thm:quest-NPh-kac}
   \INCLUSq{} is \NP-complete even if $k=1$
   and there is only one row containing $\blank$\hy entries.
 \end{theorem}
 \begin{proof}
   We give a polynomial-time reduction from \probfont{Closest String},
   which is well-known to be \NP-hard even for binary alphabets~\cite{litman}. Let $(S,r)$ with $S=(s_1,\dotsc,s_n)$ and $s_i \in
   \{0,1\}^L$ for every $i \in [n]$ be the given instance of
   \probfont{Closest String}. Then the set $M$ of vectors contains one
   vector $s_i$ for every $i \in [n]$ and additionally the vector
   $\vec{q}=\{\blank\}^L$. It is easy to observe that for
   every \YES-instance of $(M,1,r)$ there exists a solution which completes $\vec{q}$ to the closest string of $(S,r)$, and hence
   $(S,r)$ is a \YES\hy instance of \probfont{Closest String} if and
   only if $(M,1,r)$ is a \YES{}\hy instance of \INCLUSq{}.
 \end{proof}

}

\section{Going Beyond Boolean Domain}
\label{sec:boundeddomain}

In this section, we briefly discuss two generalizations of the clustering problems under consideration that allow for larger domain size, where each generalization is based on a different way of measuring distance between vectors in higher domains. In particular, we discuss the Hamming distance and the Manhattan distance over a domain $Q=\{0,1,\dots,q-1,\blank\}$, for some $q\geq 2$.
\lv{

\pbDef{\textsc{HAM-IN-Clustering-Completion}$_q$}{A subset $M$ of $\{0,1,\dots,q-1,\blank\}^d$ and $k, r \in \Nat$.}{Is there a completion $M^*$ of $M$ and subset $S \subseteq M$ with $|S|\leq k$ such that $\HDIST(S,\vec{a})\leq r$ for every $\vec{a} \in M$?}
 
\pbDef{\textsc{MAN-IN-Clustering-Completion}$_q$}{A subset $M$ of $\{0,1,\dots,q-1,\blank\}^d$ and $k, r \in \Nat$.}{Is there a subset $S \subseteq M$ with $|S|\leq k$ such that for every $\vec{a} \in M$ there exists $\vec{s}\in S$ such that $\sum_{t=1}^d |a[t]-s[t]|$ is at most $r$?}

The generalizations of the other problems to higher domains w.r.t.\
the Hamming and Manhattan distance, respectively, are defined
analogously. Observe that for $q=2$, the problems we obtain are
precisely those we introduced in Section~\ref{section:prelims}.

}
Our aim in this section is to extend our results from matrices over the Boolean domain to these generalizations, and the
main tools we use are two encodings of domain values.
We define the two encodings $\alpha : [q]\cup\{\blank\} \rightarrow
\{0,1,\blank\}^q$ and $\beta :[q]\cup \{\blank\} \rightarrow \{0,1,\blank\}^q$, where $\alpha(i)$
is the binary encoding of $2^i$ and $\beta(i)$ is the unary encoding
of $i$ if $i\neq \blank$ and $\alpha(i)=\beta(i)=\blank^q$, otherwise. Moreover, for $\vec{v} \in \{0,1\}^d$, we let
$\alpha(\vec{v})$ and $\beta(\vec{v})$ be the vectors in $\{0,1\}^{qd}$
obtained from $\vec{v}$ by replacing each coordinate $i\in
[d]$ with a \emph{block} of $q$ coordinates equal to $\alpha(i)$
and $\beta(i)$, respectively.
For example, if $Q=\{0,1,2,\blank\}$ and $d=2$, then
$\alpha((0,2))=(0,0,1,1,0,0)$ and $\beta((0,2))=(0,0,0,0,1,1)$.

It is easy to verify that there is a direct correspondence between the
vector distances in a matrix $M$ over $Q^d$ and the Hamming vector distances in the
matrix over $\{0,1,\blank\}^{qd}$ obtained by applying the respective
encoding function to $M$.
\begin{observation}
  \label{obs:trans}
  For each $\vec{a},\vec{b} \in Q^d$ it holds that
  $\HDIST(\vec{a},\vec{b})\cdot 2=\HDIST(\alpha(\vec{a}),\alpha(\vec{b}))$ and that $\sum_{t=1}^d
  |a[t]-b[t]|=\HDIST(\beta(\vec{a}),\beta(\vec{b}))$.
\end{observation}

\lv{
For each $i\in [d]$, we will call the set of coordinates $\{(i-1)\cdot
q+1,(i-1)\cdot q+2,\dots, (i-1)\cdot q+q$ a \emph{block}. }
Consider a
matrix $M$ obtained by applying $\alpha$ (or $\beta$) to a matrix
$M'$. A completion $M^*$ of $M$ is \emph{block-preserving}
w.r.t.~$\alpha$ (respectively $\beta$) if for each
vector $\vec{v}\in M^*$ the $i$-th block of $\vec{v}$ is equal to
$\alpha(i)$ (respectively $\beta(i)$) for some $i \in Q$.
Equivalently, $M^*$ is block-preserving w.r.t. $\alpha$ (or $\beta$) if it can be
obtained by applying $\alpha$ (or $\beta$, respectively) to the
elements of some completion of the matrix $M'$.

For $\textsc{Prob}\in \IADCLUSq$, let $\textsc{Prob}_\alpha$ and $\textsc{Prob}_\beta$ be
the adaptation of $\textsc{Prob}$ to the case where we additionally
require the completion $M^*$ of $M$ to be block-preserving
(w.r.t. $\alpha$ or $\beta$).
Since both encodings only increase the dimension of the vectors by a
constant factor, Observation~\ref{obs:trans} allows us to reduce the
completion problems over $Q$ to the question of finding
block-preserving completions of Boolean matrices. It is
easy to argue that all the developed algorithmic techniques can be
extended to the block-preserving variants of the problems.
\lv{
For
instance, finding and removing irrelevant vectors is not affected by
blocks. Additionally, finding and removing irrelevant coordinates is safe as
long as one always treats all coordinates of a block in the
same manner. Moreover, when we need to consider a completion of
certain $\blank$-entries, we will only consider the completions that are
block-preserving.

\begin{corollary}
  Let $\textsc{Prob}\in
  \IADCLUSq{}$ with parameterization $\iota\subseteq \{k,r,\comb\}$:
  \begin{itemize}
  \item If $\textsc{Prob}_\alpha$ is \FPT{} (or \XP) parameterized by $\iota$, then so is $\textsc{HAM-Prob}_q$.
  \item If $\textsc{Prob}_\beta$ is \FPT (or \XP) parameterized by $\iota$, then so is $\textsc{MAN-Prob}_q$.
  \end{itemize}
\end{corollary}
In other words, all our
\FPT{}-results and \XP{}-results also carry over to the finite
domain case.
}
\sv{
As a corollary, we obtain that all our
\FPT{}-results and \XP{}-results also carry over to the finite
domain case.
}

\section{Conclusion}
\label{sec:concl}

We provided a systematic study of the parameterized complexity of
fundamental clustering problems for incomplete data.
Our results draw a detailed map of the complexity landscape for the studied problems
and showcase a sharp contrast between the settings that are
fixed-parameter tractable and those which are not.

Finally, we believe that the insights and
techniques showcased in this paper are of general interest. Indeed, in
essence they show that vectors over a bounded domain which are packed
in dense clusters have non-trivial combinatorial properties that only
become accessible through a suitable set representation. We hope that
these insights and techniques turn out to be useful in other
settings as well.

\section{Acknowledgements}

Robert Ganian acknowledges support by the Austrian Science Fund (FWF,
projects P31336 and Y 1329).
Stefan Szeider acknowledges support by the Austrian Science Fund (FWF,
project P32441)
and
the Vienna Science and Technology Fund (WWTF, project ICT19-065).
Sebastian Ordyniak acknowledges support from the
Engineering and Physical Sciences Research Council (EPSRC, project EP/V00252X/1).

\bibliography{literature}

\end{document}